\documentclass[aps,pra,%
final,%
notitlepage,%
oneside,%
twocolumn,
nobibnotes,%
nofootinbib,%
superscriptaddress,%
centertags,
floatfix]%
{revtex4-1}
\usepackage{graphicx}
\usepackage{amsmath,amsthm,amssymb,amscd}
\DeclareMathOperator{\Tr}{Tr}
\DeclareMathOperator{\sinc}{sinc}

\newtheorem{theorem}{Theorem}
\usepackage{color}
\usepackage{hyperref}
\usepackage{siunitx} 

\begin{document}

\title{Bloch-Messiah reduction for twin beams of light}

\author{D.~B.~Horoshko}\email{Dmitri.Horoshko@univ-lille.fr}
\affiliation{Univ. Lille, CNRS, UMR 8523 - PhLAM - Physique des Lasers Atomes et Mol\'ecules, F-59000 Lille, France}
\affiliation{B.~I.~Stepanov Institute of Physics, NASB, Nezavisimosti Ave.~68, Minsk 220072 Belarus}%

\author{L.~La Volpe}
\affiliation{Laboratoire Kastler Brossel, UPMC-Sorbonne Universit\'e, CNRS, ENS-PSL Research University,
Coll\`ege de France, 4 place Jussieu, 75252 Paris, France}

\author{F.~Arzani}
\affiliation{Universit\'e de Lorraine, CNRS, Inria, LORIA, F-54000 Nancy, France}
\affiliation{LIP6, CNRS, Sorbonne Universit\'e, 4 place Jussieu, 75005 Paris, France}

\author{N.~Treps}
\affiliation{Laboratoire Kastler Brossel, UPMC-Sorbonne Universit\'e, CNRS, ENS-PSL Research University,
Coll\`ege de France, 4 place Jussieu, 75252 Paris, France}

\author{C.~Fabre}
\affiliation{Laboratoire Kastler Brossel, UPMC-Sorbonne Universit\'e, CNRS, ENS-PSL Research University,
Coll\`ege de France, 4 place Jussieu, 75252 Paris, France}

\author{M.~I.~Kolobov}
\affiliation{Univ. Lille, CNRS, UMR 8523 - PhLAM - Physique des Lasers Atomes et Mol\'ecules, F-59000 Lille, France}

\date{\today}

\begin{abstract}
We study the Bloch-Messiah reduction of parametric downconversion of light in the pulsed regime with a nondegenerate phase matching providing generation of twin beams. We find that in this case every squeezing eigenvalue has multiplicity at least two. We discuss the problem of ambiguity in the definition of the squeezing eigenmodes in this case and develop two approaches to unique determination of the latter. First, we show that the modal functions of the squeezing eigenmodes can be tailored from the Schmidt modes of the signal and idler beams. Alternatively, they can be found as a solution of an eigenvalue problem for an associated Hermitian squeezing matrix. We illustrate the developed theory by an example of frequency non-degenerate collinear twin beams generated in beta barium borate crystal. On this example we demonstrate how the squeezing eigenmodes can be approximated analytically on the basis of the Mehler's formula, extended to complex kernels. We show how the multiplicity of the eigenvalues and the structure of the eigenmodes are changed when the phase matching approaches the degeneracy in frequency.
\end{abstract}

\maketitle

\section{Introduction}

In the process of parametric downconversion (PDC) one photon of pump wave is converted into a pair of signal and idler photons, which are almost simultaneous \cite{Zeldovich69,Burnham70}. A multimode analysis \cite{KlyshkoBook,Hong85} of this process shows that the correlation time between the two photons is determined by the inverse of the downconverted light bandwidth, which is typically in the sub-picosecond range. Thus, detection of the idler photon results in a localization of the signal one \cite{Zeldovich69}, which is known as ``photon heralding'' technique \cite{Hong86}, widely used today. Analysis of the joint quantum state of two photons shows that a discrete set of orthogonal modes, Schmidt modes, can be defined for each photon of the pair such that detection of a photon in an idler Schmidt mode projects the other photon onto the corresponding signal Schmidt mode \cite{Law00,Grice01}. In the high-gain regime, when the initial spontaneous photon pairs experience amplification by stimulated emission and many photons are generated in the principal modes, nondegenerate PDC results in generation of two beams of light, signal and idler, having equal numbers of photons in any time interval longer than the inverse bandwidth of one beam. These highly correlated beams are known as ``twin beams'' and were generated both in the cavity \cite{Heidmann87} and the single-pass \cite{Aytur90} configurations. Early multimode theory of PDC \cite{Kolobov89,Reid89} allowed one to calculate many important properties of the generated field for the case of monochromatic pump. Later a modal decomposition for pulsed high-gain PDC was introduced \cite{Bennink02}, which represented the output field as an assembly of independent squeezing eigenmodes, each being in a single-mode squeezed state. Essentially similar approach was formulated in the language of symplectic transformations \cite{Braunstein05,McKinstrie13,Cariolaro16a,Cariolaro16b}, and is known as Bloch-Messiah reduction in reference to a similar reduction in the elementary particle physics \cite{Bloch62,Balian65}. In the case of degenerate PDC, where the signal and the idler photons are indistinguishable, the Bloch-Messiah reduction has proven to be a powerful tool for determining the squeezing eigenmodes for single-pass optical parametric amplifiers (OPA) \cite{Wasilewski06,Lvovsky07} and multi-pass optical parametric oscillators \cite{Patera10,Patera12,PateraThesis,Jiang12}. Application of this formalism to a degenerate PDC with a monochromatic pump, undertaken recently by some of us \cite{Lipfert18}, resulted in a successful identification of bichromatic squeezing eigenmodes. For nondegenerate high-gain PDC it has been found \cite{Migdal10} that up to certain level of gain the Schmidt modes of photon pairs determine a modal decomposition for the signal and idler beams such that the corresponding modes are in a \emph{two-mode} squeezed state. This decomposition is obviously related to the Bloch-Messiah reduction but does not coincide with it, since in the latter, as introduced by Braunstein \cite{Braunstein05}, the field is represented as an assembly of \emph{single-mode} squeezed states.

The aim of the present article is to apply the Bloch-Messiah reduction to twin beams generated in a non-degenerate PDC with any pump, pulsed or continuous wave. We develop a general formalism applicable in all cases where signal and idler beams can be discriminated either by frequency, direction, or polarization. We find a fundamental result applicable to all these cases: every squeezing eigenvalue of twin beams has multiplicity at least two. This result follows directly from the symmetry of the interaction Hamiltonian in the case where the downconverted light is partitioned into a signal and an idler parts. As consequence, the squeezing eigenmodes are determined up to an orthogonal rotation in the subspace, corresponding to a given eigenvalue. Such a rotation leads to ambiguity in the definition of squeezing eigenmodes and often makes impossible obtaining reproducible results when the eigenmodes are found numerically by a singular value decomposition (SVD) of the Bogoliubov transformation matrices \cite{Braunstein05} or by a Takagi factorization of the transformation generator matrix \cite{Bennink02,Cariolaro16a,Cariolaro16b}. We propose two solutions for this problem. First, we show how the squeezing eigenmodes can be composed from the Schmidt modes of each beam. Second, we show that the Takagi factorization of the squeezing matrix can be easily found from the spectral decomposition of an associated Hermitian matrix. As consequence, we find that under rather simple experimental conditions the modal functions of the squeezing eigenmodes can be identified as transform-limited waveforms. These modal functions create a basis in the subspace, corresponding to a given eigenvalue, and are represented  (up to constant phase) by \emph{real} functions of their spatio-temporal arguments. Measurement of pulsed squeezing is realized by pulsed homodying \cite{Medeiros14}, which requires a precise shaping of the local oscillator pulse to fit the target eigenmode. Thus, finding a simple transform-limited basis in a continuous family of eigenfunctions means a significant simplification of the measurement process, especially in the full three-dimensional case, where the modal functions depend on two transversal coordinates and time. We illustrate our general results by an example where the modal functions depend on frequency only, postponing the treatment of the much more complicated three-dimensional case to a future publication.

The article is structured as follows. In Sec.~\ref{symplectic} we give a brief review of the complex symplectic formalism for description of Gaussian unitary transformations, encompassing all variants of PDC with undepleted pump. The Bloch-Messiah reduction has a simple mathematical form in terms of this formalism. This formalism simplifies the mathematics of subsequent sections. Besides, we plan to apply this formalism in further publications to the description of three-dimensional spatio-temporal modes of a noncollinear OPA, which explains its rather detailed character. In Sec.~\ref{eigenmodes} we review different ways of determining the squeezing eigenmodes and show their equivalence. We consider here the Magnus expansion of the field evolution operator. The complex symplectic formalism helps us to clarify the general structure of the lowest orders of the Magnus expansion and the role of the interaction picture reference frame. Section~\ref{twin} is devoted to generation of twin beams, which is a special case of PDC where the signal and the idler modes are distinct. We show that in this special case the Hamiltonian of the unitary transformation possesses a symmetry, leading to double multiplicity of the squeezing eigenvalues. We describe two procedures allowing one to avoid the ambiguity in the determination of the squeezing eigenmodes. The developed theory is illustrated in Sec.~\ref{example} by a concrete example of bright twin beams discriminated in frequency only with the conditions close to the experiment, which is underway in the Sorbonne University group. Section~\ref{conclusions} summarizes the results and makes an outlook for the future research.

\section{Symplectic  formalism for Gaussian unitary transformations}\label{symplectic}

\subsection{Field decomposition}

In quantum theory a multimode optical field is described by a Hermitian operator $A(\mathbf{r},t)$ of its vector potential in point $\mathbf{r}$ and time $t$. In the absence of sources this operator, in the Coulomb gauge, can be uniquely decomposed into a complete orthonormal set of modal functions $\{f_k(\mathbf{r},t)\}$ \cite{GrynbergBook,Fabre08}:
\begin{equation}\label{A}
 A(\mathbf{r},t) = \sum_k a_k f_k^*(\mathbf{r},t) + \sum_k a_k^\dagger f_k(\mathbf{r},t),
\end{equation}
so that $-\partial_t A(\mathbf{r},t) = E(\mathbf{r},t)$ is electric component and $\nabla\times A(\mathbf{r},t) = B(\mathbf{r},t)$ is the magnetic component of the optical field. The operators $a_k$ and $a_k^\dagger$ are known as photon annihilation and creation operators for the $k$th mode and obey the canonical commutation relations $[a_k, a_l^\dagger]=\delta_{kl}$. In the presence of sources these operators become slowly varying functions of time or one of spatial coordinates.

In the following we consider a finite number $n$ of optical modes. We introduce a column vector of field operators $\mathbf{a}=(a_1,...a_n,a_1^\dagger,...,a_n^\dagger)^T$, and a column vector of modal functions $\mathbf{f}(\mathbf{r},t)=(f_1(\mathbf{r},t),...,f_n(\mathbf{r},t),f_1^*(\mathbf{r},t),...f_n^*(\mathbf{r},t))^T$. Now Eq.~(\ref{A}) can be rewritten in a matrix form as  $A(\mathbf{r},t) = \mathbf{f}^\dagger(\mathbf{r},t)\mathbf{a}$, and the commutation relations as $\mathbf{a}\mathbf{a}^\dagger-\left(\mathbf{a}^{\dagger T}\mathbf{a}^T\right)^T=\mathbf{K}$, where
\begin{equation}\label{K}
 \mathbf{K} = \left(\begin{array}{cc} \mathbb{I} & 0 \\ 0 & -\mathbb{I} \end{array}\right)
\end{equation}
and $\mathbb{I}$ is the $n\times n$ unit matrix. Here and below we adopt the bold font for $2n\times 2n$ matrices and $2n\times 1$ column vectors, the uppercase letters being reserved for matrices and the lowercase ones for vectors. The ``blackboard bold'' font is used for $n\times n$ matrices, while $n\times 1$ column vectors are marked by a bar above a lowercase letter. We introduce here a ``complex symplectic matrix formalism'', where a Hermitian conjugation applied to a matrix of operators means a transposition of the matrix and a Hermitian conjugation of its elements. Transposition applied to such a matrix means a transposition of the matrix without any effect to the elements. This formalism differs from the traditional real symplectic formalism \cite{Simon94,ParisBook,Adesso14,Weedbrook12} in two aspects. First, our variables are non-Hermitian operators of photon creation and ahhihilation and not positions and momenta. Such an approach is mentionned in Refs.~\cite{Simon94,ParisBook,Adesso14}, though not used for practical calculations. Second, all relations are expressed in a compact matrix form, simplifying the calculation, while in the traditional approach many important relations are expressed via matrix elements.

\subsection{General Gaussian unitary transformation}\label{general}

Interaction between the modes in a nonlinear optical process results in a transformation of the field operators. The case of a linear in $\mathbf{a}$ transformation is called a Gaussian unitary transformation and has the form
\begin{equation}\label{S}
 \mathbf{a}' = \mathbf{S}\mathbf{a},
\end{equation}
where $\mathbf{S}$ is a complex matrix. The conservation of the commutation relations requires $\mathbf{S}\mathbf{K}\mathbf{S}^\dagger=\mathbf{K}$, i.e. that the matrix $\mathbf{S}$ is complex symplectic, and has the following structure \cite{Simon94,ParisBook,Adesso14}:
\begin{equation}\label{Sstructure}
 \mathbf{S}= \left(\begin{array}{cc} \mathbb{S}_0 & \mathbb{S}_I \\ \mathbb{S}_I^* & \mathbb{S}_0^* \end{array}\right),
\end{equation}
where the complex $n\times n$ matrices $\mathbb{S}_0$ and $\mathbb{S}_I$, known as Bogoliubov transfromation matrices, satisfy the relation $\mathbb{S}_0\mathbb{S}_0^\dagger -\mathbb{S}_I\mathbb{S}_I^\dagger =\mathbb{I}$ and the matrix $\mathbb{S}_0\mathbb{S}_I^T$ is complex symmetric.

For each symplectic matrix $\mathbf{S}$ there exists a unitary operator $\mathcal{U}$, defined up to phase, producing the corresponding transformation of the field operators \cite{Simon94}
\begin{equation}\label{U}
 \mathbf{a}' = \mathcal{U}^\dagger\mathbf{a}\,\mathcal{U}.
\end{equation}
For two subsequent transformations we have $\mathbf{S}=\mathbf{S}_2\mathbf{S}_1$ and  $\mathcal{U}=\mathcal{U}_2\,\mathcal{U}_1$.

Any unitary operator can be written as $\mathcal{U} = \exp(-i\mathcal{H})$, where the Hermitian operator $\mathcal{H}$ is called transformation generator. For a Gaussian unitary transformation its generator is a polynomial of $\mathbf{a}$ of the second order at maximum. For the sake of simplicity we omit the linear part of this dependence, which is rather trivial, and consider the generators being quadratic forms in the field operators: $\mathcal{H} = \frac12\mathbf{a}^\dagger \mathbf{H} \mathbf{a}$, where the Hermitian matrix $\mathbf{H}$, which we call ``transformation generator matrix'', has the structure
\begin{equation}\label{H}
 \mathbf{H} = \left(\begin{array}{cc} \mathbb{H}_0 & \mathbb{H}_I \\ \mathbb{H}_I^* & \mathbb{H}_0^* \end{array}\right)
\end{equation}
with $\mathbb{H}_0$ Hermitian and $\mathbb{H}_I$ complex symmetric.

The symplectic matrix, Eq.~(\ref{Sstructure}), can also be expressed via its generator \cite{Adesso14}, which is simply related to the matrix of the quadratic form: $\mathbf{S} = \exp(-i\mathbf{K}\mathbf{H})$. Thus, any symplectic matrix can be written as
\begin{equation}\label{SviaH}
 \mathbf{S} = \exp\left(\begin{array}{cc} -i\mathbb{H}_0 & -i\mathbb{H}_I \\ i\mathbb{H}_I^* & i\mathbb{H}_0^* \end{array}\right).
\end{equation}

The dimension of the group, created by symplectic matrices of the form Eq.~(\ref{Sstructure}) or Eq.~(\ref{SviaH}) is $2n^2+n$. This is also the number of linearly independent generators of the form $\mathbf{K}\mathbf{H}$. We note that any pair of $n\times n$ matrices, a Hermitian $\mathbb{H}_0$  and a complex symmetric $\mathbb{H}_I$, define an $n$-dimensional Gaussian unitary transformation. Thus, parametrisation of a Gaussian unitary transformation in terms of these matrices is much simpler than in terms of the Bogoliubov transformation matrices $\mathbb{S}_0$ and $\mathbb{S}_I$, which should satisfy additional equations involving both matrices, as indicated above.

\subsection{Passive Gaussian unitary transformation}\label{passive}

An important subclass of Gaussian unitary transformations contains the transformations, preserving the total number of photons $\mathbf{a}^\dagger \mathbf{a}$. Such transformations correspond physically to mixing different modes on a multiport interferometer and are known as passive Gaussian unitary transformations \cite{Simon94,ParisBook,Weedbrook12,Adesso14}. The passive transformations necessarily have $\mathbb{H}_I=\mathbb{S}_I=0$, the matrix $\mathbb{S}_0$ becoming unitary: $\mathbb{S}_0=\exp(-i\mathbb{H}_0)$. The symplectic matrix
\begin{equation}\label{Spassive}
 \mathbf{S} = \left(\begin{array}{cc} e^{-i\mathbb{H}_0} & 0 \\ 0 & e^{i\mathbb{H}_0^*} \end{array}\right)
\end{equation}
is also unitary in this case. The corresponding transformation generator can be obtained by writing
\begin{eqnarray}\label{Hpassive}
\mathcal{H} &=& \frac12\mathbf{a}^\dagger \mathbf{H} \mathbf{a} = \frac12\left(\begin{array}{cc} \bar a^\dagger & \bar a^{T} \end{array}\right)\left(\begin{array}{cc} \mathbb{H}_0 & 0 \\ 0 & \mathbb{H}_0^* \end{array}\right)\left(\begin{array}{c} \bar a \\ \bar a^{\dagger T} \end{array}\right)\\\nonumber
&=& \frac12\left(\bar a^\dagger \mathbb{H}_0\bar a + \bar a^T \mathbb{H}_0^*\bar a^{\dagger T}\right) = \bar a^\dagger \mathbb{H}_0\bar a + \frac12\Tr \mathbb{H}_0,
\end{eqnarray}
where $\bar a = (a_1,...a_n)^T$ is the column vector of the annihilation operators of the modes.
The corresponding unitary operator in the state space is (up to phase)
\begin{equation}\label{Upassive}
 \mathcal{U} = e^{-i\bar a^\dagger \mathbb{H}_0\bar a}.
\end{equation}

The dimension of this (compact) subgroup of transformations is $n^2$.

\subsection{Mode-wise squeezing}

Another important subclass of Gaussian unitary transformations consists of single-mode squeezing of all modes with the squeezing parameters $\{r_1,...,r_n\}$, some of which may be zero \cite{Simon94,ParisBook,Weedbrook12,Adesso14}. For definiteness we accept that all squeezing parameters are non-negative and sorted in the decreasing order. They can be written in a form of a positive diagonal matrix $\mathbb{R}=\mathrm{diag}\{r_1,...,r_n\}$. We define the mode-wise squeezing transformation as a Gaussian unitary transformation with $\mathbb{H}_0=0$ and $\mathbb{H}_I=i\mathbb{R}$. The symplectic matrix for this class of transformations is
\begin{equation}\label{Smodewise}
 \mathbf{S} = \exp\left(\begin{array}{cc} 0 & \mathbb{R} \\ \mathbb{R} & 0 \end{array}\right)
 = \left(\begin{array}{cc} \cosh(\mathbb{R}) & \sinh(\mathbb{R}) \\ \sinh(\mathbb{R}) & \cosh(\mathbb{R}) \end{array}\right).
\end{equation}

The corresponding Hamiltonian can be obtained by writing %
\begin{eqnarray}\label{Hmw}
\mathcal{H} &=& \frac12\mathbf{a}^\dagger \mathbf{H} \mathbf{a} = \frac12\left(\begin{array}{cc} \bar a^\dagger & \bar a^{T} \end{array}\right)\left(\begin{array}{cc} 0 & i\mathbb{R} \\ -i\mathbb{R} & 0 \end{array}\right) \left(\begin{array}{c} \bar a \\ \bar a^{\dagger T} \end{array}\right)\\\nonumber
&=& \frac12\left(i\bar a^\dagger \mathbb{R}\bar a^{\dagger T} - i\bar a^T \mathbb{R}\bar a\right).
\end{eqnarray}
The corresponding unitary operator in the state space is
\begin{equation}\label{Umw}
 \mathcal{U} =  e^{\frac12\left(\bar a^\dagger \mathbb{R}\bar a^{\dagger T} - \bar a^T \mathbb{R}\bar a\right)}.
\end{equation}

The dimension of this (noncompact) subgroup of transformations is $n$.

\subsection{Bloch-Messiah reduction}

The central mathematical procedure for determining the modes of squeezing of a multimode optical field is the decomposition of an arbitrary Gaussian unitary transformation into two passive transformations and one mode-wise squeezing transformation:
\begin{equation}\label{BMR}
 \mathbf{S} = \left(\begin{array}{cc} \mathbb{V} & 0 \\ 0 & \mathbb{V}^* \end{array}\right) \exp\left(\begin{array}{cc} 0 & \mathbb{R} \\ \mathbb{R} & 0 \end{array}\right) \left(\begin{array}{cc} \mathbb{Q} & 0 \\ 0 & \mathbb{Q}^* \end{array}\right)^\dagger,
\end{equation}
where the matrices $\mathbb{V}$ and $\mathbb{Q}$ are unitary and the matrix $\mathbb{R}$ is positive diagonal. This decomposition is known as Bloch-Messiah reduction. It was introduced by Bloch and Messiah for fermions \cite{Bloch62} and later generalized to bosons \cite{Balian65}. The physical meaning of this procedure is the possibility of realizing any Gaussian unitary transformation by means of two multiport interferometers and a number of single-mode squeezers \cite{Braunstein05}. It should be noted that in practice a passive Gaussian transformation can be realized not by a multiport interferometer, but by a change of modal basis, corresponding to a proper shaping of the local oscillator used in the homodyne measurement of the field.

Decomposition defined by Eq.~(\ref{BMR}) can be written in the following form, used by Braunstein \cite{Braunstein05}:
\begin{equation}\label{BMRAB}
 \mathbb{S}_0 = \mathbb{V}\cosh(\mathbb{R})\mathbb{Q}^\dagger, \,\, \mathbb{S}_I = \mathbb{V}\sinh(\mathbb{R})\mathbb{Q}^T,
\end{equation}
which means a simultaneous SVD of two Bogoliubov transformation matrices.

The corresponding expression for the unitary operator in the state space is
\begin{equation}\label{BMRU}
 \mathcal{U} = e^{-i\bar a^\dagger \mathbb{H}_{V}\bar a} e^{\frac12\left(\bar a^\dagger \mathbb{R}\bar a^{\dagger T} - \bar a^T \mathbb{R}\bar a\right)} e^{i\bar a^\dagger \mathbb{H}_{Q}\bar a},
\end{equation}
where the Hermitian matrices $\mathbb{H}_V$ and $\mathbb{H}_Q$ are the generators of the corresponding unitary matrices in Eq.~(\ref{BMR}): $\mathbb{V}=e^{-i\mathbb{H}_V}$, $\mathbb{Q}=e^{-i\mathbb{H}_Q}$.

\subsection{Gaussian states}\label{Gaussian}

A fundamental property of a Gaussian transformation is that it transforms a Gaussian state into another Gaussian one. Any $n$-mode state with a density operator $\rho$ is characterized by its Wigner function
\begin{equation}
W(\boldsymbol\alpha)=\frac1{\pi^{2n}} \int d^2\bar\lambda\Tr\left\{\rho e^{(\bar a^\dagger-\bar\alpha^\dagger)\bar\lambda-\bar\lambda^\dagger(\bar a-\bar\alpha)}\right\},
\end{equation}
where $\bar\alpha = (\alpha_1,...,\alpha_n)^T$ and $\bar\lambda = (\lambda_1,...,\lambda_n)^T$ are two column-vectors of complex variables and $\boldsymbol\alpha=(\bar\alpha^T, \bar\alpha^\dagger)^T$.

A state is called Gaussian if its Wigner function is a multivariate Gaussian distribution \cite{Simon94,ParisBook,Adesso14}:
\begin{equation}
W(\boldsymbol{\alpha})=\frac1{(2\pi)^{n}\sqrt{\det\boldsymbol{\Sigma}}} e^{-\frac12(\boldsymbol\alpha^\dagger-\boldsymbol\alpha^\dagger_0) \boldsymbol\Sigma^{-1} (\boldsymbol\alpha-\boldsymbol\alpha_0)},
\end{equation}
where $\boldsymbol\alpha_0 = \langle\mathbf{a}\rangle$ is the column vector of mean field and
\begin{equation}
\boldsymbol\Sigma = \frac{
\langle(\mathbf{a}-\boldsymbol\alpha_0) (\mathbf{a}^\dagger-\boldsymbol\alpha_0^\dagger)\rangle}2 + \frac{\langle\left[(\mathbf{a}^\dagger-\boldsymbol\alpha_0^\dagger)^T (\mathbf{a}-\boldsymbol\alpha_0)^T\right]^T \rangle}2
\end{equation}
is the $2n\times 2n$ complex covariance matrix.

If a Gaussian state with mean $\boldsymbol\alpha_0$ and covariance matrix $\boldsymbol\Sigma$ undergoes a Gaussian transformation with matrix $\mathbf{S}$, then the state at the output is a Gaussian state with the mean $\mathbf{S}\boldsymbol\alpha_0$ and the covariance matrix $\boldsymbol\Sigma' = \mathbf{S}\boldsymbol\Sigma\mathbf{S}^\dagger$ \cite{Simon94}.

The complex covariance matrix $\boldsymbol\Sigma$ has the following structure
\begin{equation}\label{Sigma}
\boldsymbol\Sigma = \left(\begin{array}{cc} \Sigma_0   & \Sigma_I \\
                                            \Sigma_I^* & \Sigma_0^* \end{array}\right),
\end{equation}
where the Hermitian matrix
\begin{equation}\label{Sigma0}
\Sigma_0= \frac12\left\langle(\bar{a}-\bar\alpha_0)(\bar{a}^\dagger-\bar\alpha_0^\dagger) +\left[(\bar{a}^\dagger-\bar\alpha_0^\dagger)^T(\bar{a}-\bar\alpha_0)^T\right]^T\right\rangle
\end{equation}
is the phase-insensitive covariance (also known as coherency matrix), while the complex symmetric matrix
\begin{equation}\label{SigmaI}
\Sigma_I= \left\langle(\bar{a}-\bar\alpha_0)(\bar{a}-\bar\alpha_0)^T\right\rangle
\end{equation}
is the phase-sensitive covariance of the field (also known as anomalous correlator \cite{Kilin89}).

In the next section we will see how the developed formalism leads to a natural determination of squeezing eigenmodes in the process of PDC.

\section{Squeezing eigenmodes}\label{eigenmodes}

\subsection{Determining the squeezing eigenmodes by Bloch-Messiah reduction}\label{definition}

In many nonlinear optical experiments, such as PDC and four-wave mixing (FWM), a squeezed light is generated by amplification of vacuum fluctuations. The formalism of Bloch-Messiah reduction, when applied to the vacuum field at the input, gives a possibility to define a set of modes at the output, such that each mode of the field is in a squeezed vacuum state with the squeezed quadrature along the same direction in the phase space, or vacuum. Indeed, the field transformation Eq.~(\ref{S}) can be written with the help of Eq.~(\ref{BMR}) as
\begin{equation}
\mathbf{V}^\dagger\mathbf{a}' = \left(\begin{array}{cc} \cosh(\mathbb{R}) & \sinh(\mathbb{R}) \\ \sinh(\mathbb{R}) & \cosh(\mathbb{R}) \end{array}\right) \mathbf{a}_{vac},
\end{equation}
where $\mathbf{a}_{vac} = \mathbf{Q}^\dagger\mathbf{a}$ and the unitary matrices $\mathbf{V}$ and $\mathbf{Q}^\dagger$ are the leftmost and the rightmost ones in the right hand side of Eq.~(\ref{BMR}), defined as
\begin{equation}\label{VQ}
 \mathbf{V} = \left(\begin{array}{cc} \mathbb{V} & 0 \\ 0 & \mathbb{V}^* \end{array}\right),\,\,\,
 \mathbf{Q} = \left(\begin{array}{cc} \mathbb{Q} & 0 \\ 0 & \mathbb{Q}^* \end{array}\right).
\end{equation}

The modes characterized by operators $\mathbf{b} = \mathbf{V}^\dagger\mathbf{a}'$ with the modal functions $\mathbf{g}(\mathbf{r},t) = \mathbf{V}^\dagger\mathbf{f}(\mathbf{r},t)$ are generally known as ``modes of squeezing'' or ``squeezed modes''. We shall call them ``squeezing eigenmodes'' and the corresponding diagonal values of the matrix $\mathbb{R}$ ``squeezing eigenvalues''. These modes are important for considering encoding quantum information into continuous variables of the optical field. For example, in  homodyne detection, for observing the maximal squeezing, the modal function of the local oscillator should match the mode with the maximal squeezing eigenvalue $r_1$.

\subsection{Determining the squeezing eigenmodes by squeezing matrix}

An alternative way of determining the squeezing eigenmodes can be obtained from Eq.~(\ref{BMRU}), which we rewrite as \cite{Ma90}
\begin{equation}\label{BMRU2}
 \mathcal{U} = e^{\frac12\left(\bar a^\dagger \Gamma\bar a^{\dagger T} - \bar a^T \Gamma^*\bar a\right)} e^{-i\bar a^\dagger \mathbb{H}_{V}\bar a} e^{i\bar a^\dagger \mathbb{H}_{Q}\bar a},
\end{equation}
where $\Gamma = \mathbb{V}\mathbb{R}\mathbb{V}^T$ is a compex symmetric matrix, called ``the squeezing matrix'' and we have used the identity
\begin{equation}
e^{-i\bar a^\dagger \mathbb{H}_{V}\bar a} \bar a e^{i\bar a^\dagger \mathbb{H}_{V}\bar a} = e^{i\mathbb{H}_V}\bar a = \mathbb{V}^\dagger\bar a,
\end{equation}
following from the properties of a passive Gaussian transformation, described in Sec.~\ref{passive}. When the operator $\mathcal{U}$ acts on a vacuum field, as in the case of unseeded PDC, two rightmost factors in the right hand side of Eq.~(\ref{BMRU2}) have no effect, and the resulting state is determined only by the third factor, dependent on $\Gamma$.

Bennink and Boyd \cite{Bennink02} considered the Takagi factorization of the squeezing matrix $\Gamma = \mathbb{VRV}^T$ and called the modes defined by the columns of $\mathbb{V}$ the ``eigenmodes of the squeezing''. We see from the above that the Takagi factorization of $\Gamma$ results in the same matrix $\mathbb{V}$ as the Bloch-Messiah reduction, and therefore this way of determining the squeezing eigenmodes is equivalent to that of Sec.~\ref{definition}.

\subsection{Determining the squeezing eigenmodes by covariance matrix}

One more alternative way for determining the squeezing eigenmodes is based on the diagonalization of the covariance matrix. The state at the input of the nonlinear crystal in the case of unseeded PDC is vacuum, which is a Gaussian state with zero mean and the covariance matrix $\frac12\mathbf{I}$. In accord with Sec.~\ref{Gaussian} the output field has zero mean and the covariance matrix
\begin{equation}
\boldsymbol\Sigma' = \frac12\mathbf{S}\mathbf{S}^\dagger = \frac12\mathbf{V}e^{2\mathbf{R}}\mathbf{V}^\dagger,
\end{equation}
where
\begin{equation}
\mathbf{R} = \left(\begin{array}{cc} 0 & \mathbb{R} \\ \mathbb{R} & 0 \end{array}\right).
\end{equation}

The blocks of the covariance matrix, defined by Eq.~(\ref{Sigma}) are $\Sigma'_0=\frac12\mathbb{V}\cosh(2\mathbb{R})\mathbb{V}^\dagger$ and $\Sigma'_I=\frac12\mathbb{V}\sinh(2\mathbb{R})\mathbb{V}^T$. Note, that the matrix $\mathbb{V}$ obtained from the Takagi factorization of  $\Sigma'_I$ fits also the spectral decomposition of $\Sigma'_0$, but the inverse is not true in general.

We see thus that the squeezing eigenmodes can be obtained from the Takagi factorization of the phase-sensitive covariance $\Sigma'_I$ of the output field, similar to the approach of Shapiro and Shakeel \cite{Shapiro97}. In the latter approach the squeezing eigenmodes are obtained from a diagonalization of a real symmetric covariance matrix, including both the phase-insensitive and the phase-sensitive covariances. As we have shown above the same result can be achieved by a diagonalization of the phase-sensitive covariance alone. However this matrix is complex symmetric, not real symmetric as in Ref.~\cite{Shapiro97}.

As we have seen from the two last sections, the matrix of squeezing eigenmodes $\mathbb{V}$ can be obtained by the Takagi factorization of either the squeezing matrix $\Gamma$ determined by the physical model of PDC, or the phase-sensitive covariance matrix $\Sigma'$, directly measurable by the homodyne technique. Takagi factorization can be rather easily calculated numerically when all the squeezing eigenvalues are different. This case is typical for degenerate OPA \cite{Wasilewski06,Lvovsky07} and optical parametric oscillators \cite{Patera10,Jiang12,Arzani18}. As we will see in Sec.~\ref{twin}, in the case of twin beams each squeezing eigenvalue has a multiplicity of at least two and the squeezing eigenmodes are defined with some degree of freedom. Takagi factorization in this case requires computing a balancing matrix and its square root \cite{Cariolaro16b}, which introduces a higher level of complexity, especially in a multi-dimensional case. We show below how this complexity may be reduced in many practically important cases.

\subsection{Frequency eigenmodes}

To show how the formalism developed above is related to generation of squeezed light in PDC, we limit our consideration to the case, where the generated photons differ in one variable only, the frequency, i.e. to the case of type-I collinear PDC. However, all obtained results will be valid when the photons differ in any combination of frequency, direction and polarization with a proper relabeling of the modes.

We consider PDC in a nonlinear $\chi^{(2)}$ crystal with a pulsed plane-wave pump of central frequency $\omega_p$. We assume that the pump wave is strong enough and is undepleted. A coordinate system is chosen with the $z$-axis in the direction of the pump wave propagation and with the origin at the input edge of the crystal, see Fig.~\ref{fig:pdc}. The pump is considered as a classical wave with the (time-dependent) amplitude $E_p$, the wave vector $k_p$, and the frequency $\omega_p$. The downconverted wave may have a broad spectrum of frequencies $\omega=\omega_0+\Omega$ around the central frequency $\omega_0=\omega_p/2$, with the corresponding wave vector $k_0$.

\begin{figure}[!ht]
\centering
\includegraphics[width=0.95\linewidth]{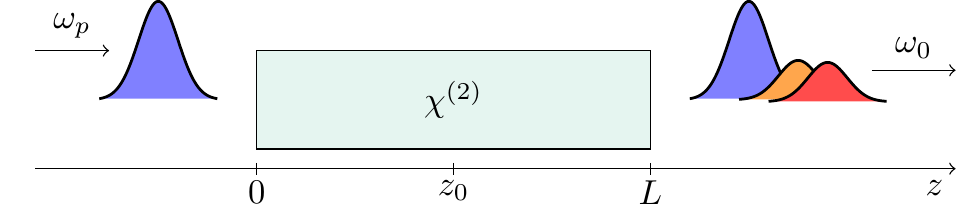}
\caption{Schematic representation of a collinear type-I PDC in the pulsed regime. Pump pulse with the central frequency $\omega_p$ passes through a $\chi^{(2)}$ nonlinear crystal. For $z>0$ this pulse is accompanied by the signal and the idler pulses at frequencies close to $\omega_0=\omega_p/2$. The point $z_0$ is the origin of the ``interaction picture reference frame'' and can be chosen anywhere inside the crystal. Relative delay between the pulses at the output is caused by the crystal positive dispersion: waves with a higher frequency travel at a lower group velocity. }
\label{fig:pdc}
\end{figure}
The down-converted wave is described by the positive-frequency operator $E^{(+)}(t,z)$ normalized to photon-flux units, which can be decomposed into Fourier components as
\begin{align}
E^{(+)}(t,z) = \frac1{2\pi}\int a(\Omega,z)e^{-i\left(\omega_0+\Omega\right)t}d\Omega,
\end{align}
where $a(\Omega,z)$ is the photon annihilation operator with the frequency $\omega_0+\Omega$ and the longitudinal coordinate $z$.

Another operator, $\epsilon(\Omega,z)$, is defined by the relation~\cite{Kolobov99}
\begin{equation}\label{aeps}
a(\Omega,z) = \epsilon(\Omega,z)e^{ik(\Omega)\left(z-z_0\right)},
\end{equation}
where $k(\Omega)$ is the wave vector of the down-converted light in the crystal, corresponding to the frequency $\omega_0+\Omega$. Operator $\epsilon(\Omega,z)$ is convenient for the description of the nonlinear interaction inside the crystal and is a quantum-mechanical analog of the classical slowly-varying amplitude~\cite{BoydBook}. The point $z_0$, where $a(\Omega,z_0) = \epsilon(\Omega,z_0)$, can be placed anywhere inside the crystal. Variation of this point results in a phase factor for the slowly varying amplitude. The transformation Eq.~(\ref{aeps}) can be alternatively considered as passage to the interaction picture \cite{Wasilewski06,Lvovsky07}. This passage is dependent on the reference frame, and the point $z_0$ can be viewed as the origin of the ``interaction picture reference frame''. We recall that the ``working reference frame'' has its origin at the crystal input, so that the field evolution in the crystal is considered between the points $z=0$ and $z=L$, where $L$ is the crystal length. Such a choice of the working reference frame simplifies the formulas. We analyze below several choices for $z_0$, including $z_0=0$ and $z_0=L/2$, used in the literature.

The classical pump field is decomposed into Fourier components as
\begin{align}
E^{(+)}_p(t,z) = \frac1{2\pi}\int \mathcal{E}^\text{in}(\Omega)e^{ik_p(\Omega)z -i\left(\omega_p+\Omega\right)t}d\Omega,
\end{align}
where $k_p(\Omega)$ is the wave vector of the pump at the frequency $\omega_p+\Omega$ and $\mathcal{E}^\text{in}(\Omega)$ is the Fourier transform of the pump wave at the crystal input at $z=0$. We note that for the type-I phase-matching the polarization of the pump wave is orthogonal to that of the generated light and the refractive index determining the dependence $k_p(\Omega)$ is different from that determining the dependence $k(\Omega)$.

Instead of working with continuous functions of frequency, we prefer to define discrete modes at frequencies, equidistantly spaced by the value $2\pi/T$, which is determined by the duration $T$ of the time interval where the quantum field is considered. We limit our consideration to the frequency band $[\omega_0-\Delta\Omega,\omega_0+\Delta\Omega]$, where the photons are mainly generated, and call the photons with positive (negative) detunings signal (idler) ones. We split each band into $m$ modes, and write the discrete frequencies as
\begin{equation}\label{Omegal}
 \Omega_l = \left(l-m-\frac12\right)\frac{2\pi}T,
\end{equation}
where $l$ runs from 1 to $n=2m$.

For these discrete modes we have a set of annihilation operators $a_l(z)=a(\Omega_l,z)$ and a set of slowly-varying amplitudes $\epsilon_l(z)=\epsilon(\Omega_l,z)$. Composing the vectors of length $2n$ we obtain $\mathbf{a}(z)=(a_1(z),...a_n(z),a_1^\dagger(z),...,a_n^\dagger(z))^T$, and $\boldsymbol{\xi}(z)=(\epsilon_1(z),...\epsilon_n(z),\epsilon_1^\dagger(z),...,\epsilon_n^\dagger(z))^T$.
		
The evolution of the downconverted light in the crystal is described by the following equation~\cite{Wasilewski06,Bennink02}:
\begin{eqnarray}\label{evolution}
 \frac\partial{\partial z}a_j(z) &=& ik(\Omega_j) a_j(z) \\\nonumber
 &-&i\sigma \sum_{l=1}^n \mathcal{E}_{j+l}^\text{in} e^{ik_p(\Omega_j+\Omega_l)z}a_l^\dagger(z),
\end{eqnarray}
where $\mathcal{E}_{j+l}^\text{in} = \mathcal{E}^\text{in}(\Omega_j+\Omega_l)$, and the first term in the right-hand side corresponds to dispersive propagation, while the second term corresponds to nonlinear interaction with the coupling constant $\sigma$, proportional to the nonlinear susceptibility of the crystal. In this nonlinear interaction a photon of the pump wave with the frequency $\Omega_j+\Omega_l$ is annihilated and converted into two photons with frequencies $\Omega_j$ and $\Omega_l$ respectively.

For slowly-varying amplitudes we obtain from Eqs.~(\ref{aeps}), (\ref{evolution})
\begin{eqnarray}\label{evolution2}
 \frac\partial{\partial z}\epsilon_j(z) = -i\sigma \sum_{l=1}^n \mathcal{E}_{j+l}^0 e^{i\Delta_{jl}(z-z_0)}\epsilon_l^\dagger(z),
\end{eqnarray}
where
\begin{equation}\label{E0}
 \mathcal{E}_{j+l}^0 = \mathcal{E}_{j+l}^\text{in}e^{ik_p(\Omega_j+\Omega_l)z_0}
\end{equation}
is the Fourier component of the pump field at $z=z_0$, and $\Delta_{jl} = k_{p}(\Omega_j+\Omega_l)- k(\Omega_j)-k(\Omega_l)$ is the phase mismatch of the corresponding modes. Equation (\ref{evolution2}) can be rewritten in a compact matrix form as
\begin{equation}\label{eq}
 \frac\partial{\partial z}\boldsymbol{\xi}(z) = -i\mathbf{F}(z) \boldsymbol{\xi}(z),
\end{equation}
where the matrix $\mathbf{F}(z)$ is given by
\begin{equation}\label{F}
 \mathbf{F}(z) = \left(\begin{array}{cc} 0 & \mathbb{F}(z) \\ -\mathbb{F}^*(z) & 0 \end{array}\right)
\end{equation}
with the complex symmetric matrix $\mathbb{F}(z)$ defined as
\begin{equation}\label{Fjl}
 \mathbb{F}_{jl}(z) =  \sigma \mathcal{E}_{j+l}^0 e^{i\Delta_{jl}(z-z_0)}.
\end{equation}

Equation~(\ref{eq}) can be written in a Hamiltonian form $\partial_z\boldsymbol{\xi}(z) = -i[\boldsymbol{\xi}(z),\mathcal{F}(z)]$ with a $z$-dependent Hamiltonian
\begin{equation}\label{Hamilt}
 \mathcal{F}(z) = \frac12\left(\bar\epsilon^\dagger \mathbb{F}(z)\bar\epsilon^{\dagger T} + \bar\epsilon^T \mathbb{F}^*(z)\bar\epsilon\right) = \frac12\boldsymbol{\xi}^\dagger \mathbf{K}\mathbf{F}(z)\boldsymbol{\xi}.
\end{equation}
where $\bar{\epsilon}=(\epsilon_1,...\epsilon_n)^T$.

Solution of Eq.~(\ref{eq}) can be written in the form of a $T$-exponent either for a symplectic matrix, transforming $\boldsymbol{\xi}(0)$ to $\boldsymbol{\xi}(L)$:
\begin{equation}\label{Texp}
 \mathbf{S} = \mathcal{T}e^{-i\int_{0}^{L}\mathbf{F}(z)dz},
\end{equation}
or for a unitary operator
\begin{equation}\label{TexpU}
 \mathcal{U} = \mathcal{T}e^{-i\int_0^L\mathcal{F}(z)dz},
\end{equation}
where $\mathcal{T}$ denotes the $z$-ordering superoperator, placing the operators (matrices) with higher $z$ to the left in the Taylor series of the exponential. Comparing Eq.~(\ref{TexpU}) with the definition of the transformation generator in Sec.~\ref{general}, we conclude that the latter is given by $\mathcal{H}=\mathcal{F}L$ in the case of $z$-independent Hamiltonian, which can be met, e.g. in the case of perfect phasematching for all modes, where $\Delta_{jl}=0$. This fact explains the widely used referring to the transformation generator $\mathcal{H}$ as ``Hamiltonian''. However, in the general case the Hamiltonian of the field is $z$-dependent and is defined by Eq.~(\ref{Hamilt}).

\subsection{Magnus expansion}\label{sec:Magnus}

Both $T$-exponents, Eqs.~(\ref{Texp}) and (\ref{TexpU}), can be represented as Magnus expansions~\cite{Magnus54,Blanes09}
\begin{eqnarray}\label{Magnus}
\mathbf{S}&=&e^{\boldsymbol\Omega_1+\boldsymbol\Omega_2+\boldsymbol\Omega_3+\dots},\\\label{MagnusU}
\mathcal{U}&=&e^{\boldsymbol\Xi_1 + \boldsymbol\Xi_2 + \boldsymbol\Xi_3 + \dots},
\end{eqnarray}
where $\boldsymbol\Omega_k$ is a $2n\times 2n$ matrix proportional to $|\sigma|^k$. The first term in Eq.~(\ref{Magnus}) is
	\begin{eqnarray}\label{Omega1}
	 \boldsymbol\Omega_1 &=& -i\int_0^L dz\, \mathbf{F}(z),
	\end{eqnarray}
and the higher order terms are expressed via integrals over nested commutators of $\mathbf{F}(z)$ with itself at different spatial points. Equivalently, $\boldsymbol\Xi_k$ is an operator proportional to $|\sigma|^k$, and the first term in Eq.~(\ref{MagnusU}) is
	\begin{eqnarray}\label{O1}
	 \boldsymbol\Xi_1 &=& -i\int_0^L dz\, \mathcal{F}(z),
	\end{eqnarray}
while the higher order terms are expressed via integrals over nested commutators of $\mathcal{F}(z)$ with itself at different spatial points.

First order Magnus approximation to the symplectic matrix is denoted as $\mathbf{S}_1$ and obtained by leaving the first term only in the exponent of Eq.~(\ref{Magnus}), while a similar approximation to the evolution operator is denoted $\mathcal{U}_1$ and obtained by leaving the first term only in the exponent of Eq.~(\ref{MagnusU}). Let us show that $\mathcal{U}_1$ is a Gaussian transformation with the corresponding symplectic matrix $\mathbf{S}_1$. Substituting  Eq.~(\ref{Hamilt}) into Eq.~(\ref{O1}) we obtain
\begin{equation}\label{O1bis}
 \boldsymbol\Xi_1 = \frac12\boldsymbol{\xi}^\dagger \mathbf{K}\boldsymbol\Omega_1\boldsymbol{\xi},
\end{equation}
corresponding to a Gaussian transformation with the symplectic matrix $e^{\mathbf{K}^2\boldsymbol\Omega_1}=e^{\boldsymbol\Omega_1}$, which is exactly $\mathbf{S}_1$.

The symplectic matrix $\mathbf{S}_1$ is determined by its generator, which we denote by $\mathbf{H}^{[1]}$, as  $ \mathbf{S}_1 = \exp\left(-i\mathbf{K}\mathbf{H}^{[1]}\right)$. The generator is determined according to Eq.~(\ref{H}) by two $n\times n$ matrices: $\mathbb{H}_0^{[1]}$, which is Hermitian, and $\mathbb{H}_I^{[1]}$, which is complex symmetric, as
\begin{equation}\label{Hk}
 \mathbf{H}^{[1]} = \left(\begin{array}{cc} \mathbb{H}_0^{[1]} & \mathbb{H}_I^{[1]} \\ \mathbb{H}_I^{[1]*} & \mathbb{H}_0^{[1]*} \end{array}\right).
\end{equation}

We have from Eqs.~(\ref{Omega1}), (\ref{F}) $\mathbb{H}_0^{[1]}=0$, and
\begin{equation}\label{HI1}
 \mathbb{H}_I^{[1]} = \int_0^L \mathbb{F}(z)dz,
\end{equation}
wherefrom with the help of Eq.~(\ref{Fjl}) we obtain
\begin{equation}\label{HIm}
 (\mathbb{H}_I^{[1]})_{jl}=\sigma L\mathcal{E}_{j+l}^0e^{i\Delta_{jl}\left(L/2-z_0\right)}\mathrm{sinc}(\Delta_{jl}L/2).
\end{equation}

From Eq.~(\ref{O1bis}) we have
\begin{equation}
\mathcal{U}_1=e^{\frac12\boldsymbol{\xi}^\dagger \mathbf{K}\boldsymbol\Omega_1\boldsymbol{\xi}} = e^{-\frac{i}2\left(\bar\epsilon^\dagger \mathbb{H}_I^{[1]}\bar\epsilon^{\dagger T} + \bar\epsilon^T \mathbb{H}_I^{[1]*}\bar\epsilon\right)}.
\end{equation}
Comparing this expression with Eq.~(\ref{BMRU2}), we arrive at the conclusion that in the first order of Magnus expansion the squeezing matrix is $\Gamma=-i\mathbb{H}_I^{[1]}$. As consequence, the squeezing eigenmodes are determined by the Takagi factorization $-i\mathbb{H}_I^{[1]} = \mathbb{VRV}^T$, where $\mathbb{V}$ is unitary. We note that the Takagi factorization for $i\mathbb{H}_I^{[1]*}$ reads $i\mathbb{H}_I^{[1]*} = \mathbb{V}^*\mathbb{RV}^\dagger$.  When the matrices $\mathbb{V}$ and $\mathbb{R}$ are found, the Bloch-Messiah reduction in the first order of Magnus expansion can be easily written in the form
\begin{eqnarray}\nonumber
 \mathbf{S}_1 &=& \exp\left(\begin{array}{cc} 0 & -i\mathbb{H}_I^{[1]} \\ i\mathbb{H}_I^{[1]*} & 0 \end{array}\right) = \exp\left(\begin{array}{cc} 0 & \mathbb{VRV}^T \\ \mathbb{V}^*\mathbb{RV}^\dagger & 0 \end{array}\right) \\\label{BMR1}
 &=& \left(\begin{array}{cc} \mathbb{V} & 0 \\ 0 &\mathbb{V}^* \end{array}\right) \exp\left(\begin{array}{cc} 0 & \mathbb{R} \\ \mathbb{R} & 0 \end{array}\right) \left(\begin{array}{cc} \mathbb{V} & 0 \\ 0 & \mathbb{V}^* \end{array}\right)^\dagger,
\end{eqnarray}
which is a special case of Eq.~(\ref{BMR}) with $\mathbb{Q}=\mathbb{V}$. Note that in this special case the coincidence of the input and the output eigenmodes reflects certain symmetry of the first order of Magnus expansion which does not hold in the higher orders.

The unitary evolution operator in the state space is given by Eq.~(\ref{BMRU}) with $\mathbb{H}_Q=\mathbb{H}_V$. When this operator acts on the vacuum, it produces a multimode squeezed state
\begin{equation}\label{state1}
  |\Psi\rangle = e^{-i\bar a^\dagger \mathbb{H}_V\bar a} e^{\frac12\left(\bar a^\dagger \mathbb{R}\bar a^{\dagger T} - \bar a^T \mathbb{R}\bar a\right)}|0\rangle.
\end{equation}

In the new basis, defined as
\begin{equation}\label{basis}
 \boldsymbol{\xi}' = \left(\begin{array}{cc} \mathbb{V} & 0 \\ 0 & \mathbb{V}^* \end{array}\right)^\dagger\boldsymbol{\xi},
\end{equation}
the transformation of the field operators has a form of mode-wise squeezing:
\begin{equation}\label{squeezing}
 \boldsymbol{\xi}_{out}' = \left(\begin{array}{cc} \cosh(\mathbb{R}) & \sinh(\mathbb{R}) \\ \sinh(\mathbb{R}) & \cosh(\mathbb{R}) \end{array}\right)\boldsymbol{\xi}_{in}'.
\end{equation}

The new modal functions are $\mathbf{f}'(\mathbf{r},t)=\mathbb{V}^\dagger\mathbf{f}(\mathbf{r},t)$. Note, that the field operator is invariant with respect to the choice of the modal basis: $A(\mathbf{r},t) = \mathbf{f}'{}^\dagger(\mathbf{r},t)\boldsymbol{\xi}' = \mathbf{f}^\dagger(\mathbf{r},t)\boldsymbol{\xi}$.

In the present work we limit our consideration to the first order of the Magnus expansion, which has been found a good approximation for not very high squeezing, below 12 dB \cite{Christ13,Lipfert18}. An analytic treatment of higher orders of the Magnus expansion in PDC can be found in Refs.~\cite{Quesada14,Quesada15,Lipfert18}.

\subsection{Reduction of the Takagi factorization to a real symmetric spectral decomposition}\label{eig}

As we have seen above, in the first order of the Magnus expansion the squeezing eigenmodes are defined by the Takagi factorization applied to the complex symmetric squeezing matrix $\Gamma=-i\mathbb{H}_I^{[1]}$, whose matrix elements are determined by Eq.~(\ref{HIm}). Takagi factorization can always be realized for a complex symmetric matrix by performing a singular value decomposition and then computing a balancing matrix \cite{Cariolaro16a,Cariolaro16b}. However, this procedure can be significantly simplified under certain conditions. We notice that the elements of the squeezing matrix can be made real if two conditions are satisfied: (i) the origin of the interaction picture frame is chosen at the center of the crystal, $z_0=L/2$, and (ii) the pump pulse at this point is transform-limited, i.e. all elements of the vector $\mathcal{E}_{j+l}^0$ have the same phase $\varphi_0$. Then, as is easily seen from  Eq.~(\ref{HIm}), all elements of the squeezing matrix have the phase $\varphi=\varphi_0+\arg(\sigma)-\pi/2$. This phase can be removed by a trivial transformation $\bar{\epsilon}\to\bar{\epsilon}e^{i\varphi/2}$ in Eq.~(\ref{eq}) which makes the squeezing matrix real.

When the squeezing matrix is real symmetric it can be represented in the form of spectral decomposition
\begin{equation}\label{orth}
 \Gamma=\mathbb{O}\Lambda \mathbb{O}^T,
\end{equation}
where $\mathbb{O}$ is a real orthogonal matrix and $\Lambda$ is the real diagonal matrix of eigenvalues.
If all eigenvalues of $\Gamma$ are non-negative, i.e. the squeezing matrix is positive semidefinite \cite{HornJohnson}, Eq.~(\ref{orth}) is a Takagi factorization, the columns of $\mathbb{O}$ are the modal functions of the modes of squeezing, and the eigenvalues of $\Gamma$ are the squeezing parameters. In the opposite case, when some of the eigenvalues of $\Gamma$ are negative, the Takagi factorization can be reconstructed from the spectral decomposition in the following way. Let $\bar e_j=(0,...,1,...0)^T$ be the $1\times n$ column vector with the $j$th element equal to 1 and all others zero. Then $\mathbb{O}\bar e_j=O_j$ is the eigenvector corresponding to the $j$th eigenvalue $\lambda_j$. We define a set of column vectors $V_j$ as follows: $V_j=O_j$ if $\lambda_j\ge0$ and $V_j=iO_j$ if $\lambda_j<0$. We build a square matrix $\mathbb{V}=(V_1,...,V_n)$ from the columns $V_j$. It is easy to verify that this matrix is unitary, since $\mathbb{O}$ is orthogonal. We define also $r_j=|\lambda_j|$ and build a diagonal matrix $\mathbb{R}$ with $r_j$ at diagonal. Now we can rewrite Eq.~(\ref{orth}) as
\begin{eqnarray}\label{orth2}
 \Gamma &=& \mathbb{O}\sum_j\lambda_j\bar e_j\bar e_j^T \mathbb{O}^T = \sum_j\lambda_jO_jO_j^T\\\nonumber
 &=& \sum_j r_jV_jV_j^T = \mathbb{VRV}^T,
\end{eqnarray}
which is a Takagi factorization for $\Gamma$. Thus, in the case of real squeezing matrix the modal functions of the squeezing eigenmodes, defined by the columns of $\mathbb{V}$, are either purely real or purely imaginary, which greatly simplifies their analysis and visualisation.

The price to be paid for the possibility to work with real modal functions is, as shown above, the necessity to have the pump pulse transform-limited at the center of the crystal. In a physical experiment the pump pulse is typically generated as transform-limited by the laser, and is therefore transform-limited at the crystal input. During its propagation inside the crystal the pulse experiences dispersion determined by the dependence $k_p(\Omega)$ in Eq.~(\ref{E0}). As a result, at the point $z_0=L/2$ it acquires the phase $k_p(\Omega)L/2$. In the quadratic dispersion approximation we can write
\begin{equation}\label{kp}
 k_p(\Omega)L/2 \approx k_{p0}L/2 + \frac12k_{p0}'L\Omega + \frac14 k_{p0}''L\Omega^2,
\end{equation}
where $k_{p0}$, $k_{p0}'$ and $k_{p0}''$ are the value and the two derivatives of $k_p(\Omega)$
at $\Omega=0$. Thus, the phase of the pump pulse at the crystal center is composed of three major terms: (i) constant phase shift due to optical oscillations at the carrier frequency, which can be included into the phase $\varphi_0$, (ii) absolute group delay of the pump pulse, which can be put to zero by adjusting the time origin, and (iii) the frequency chirp, which is equivalent to pulse spread in the temporal domain. The latter phase cannot be removed by a simple mathematical transformation and its removal has physical meaning. First, it can be neglected if the pump bandwidth $\Omega_p$ is small enough so that $\Omega_p\ll2\sqrt{2}(k_{p0}''L)^{-1/2}$. Second, it can be compensated by pre-chirping the pump pulse with the frequency chirp $-k_{p0}''L\Omega^2/4$. Since most nonlinear crystals have positive dispersion ($k_{p0}''>0$), this pre-chirp should be negative and thus require an active phase control.

In the examples of the squeezing eigenmodes presented below in Sec.~\ref{example} we suppose that one of these conditions holds and the modal functions of squeezing eigenmodes, defined in terms of the slowly-varying amplitudes $\epsilon(\Omega)$, are either real or imaginary. For the choice of $z_0=L/2$ the slowly-varying amplitude coincides with the full field $a(\Omega)$ at the crystal center, and at the crystal output the latter operator has additional phase determined by Eq.~(\ref{aeps}). This phase includes the three terms shown in Eq.~(\ref{kp}) for the pump: constant phase shift, group delay and chirp. The full analysis of the output modal functions is outside of the scope of the present treatment, which is limited to the squeezing eigenvalues and eigenmodes at the crystal center. We note that recently the eigenmodes of a similar parametric nonlinear process -- sum-frequency generation -- have been found for the case of pre-chirped pump pulse \cite{PateraSPIE2018b}.

\subsection{Summary}
In this section we have shown how the squeezing eigenmodes and the corresponding squeezing eigenvalues can be obtained from the Takagi factorization of the squeezing matrix of the multimode optical field. Further, we have shown how the squeezing matrix can be obtained from the physical model of PDC, and discussed a special case where this matrix is real. In the next section we apply this general theory to twin beams of light, having remarkable additional symmetries.

\section{Twin beams of light}\label{twin}

\subsection{Parametric downconversion in the case of twin beams}

Up to this point our consideration of PDC related any case of phase-matching: degenerate and nondegenerate. In the following sections we apply the developed formalism to the specific case of twin beams and demonstrate some important regularities inherent in this case. Twin beams are generated in non-degenerate PDC where signal and idler photons are discriminated by frequency, direction, polarization or a combination of these variables. For the description of twin beams we need to consider a set of $m$ signal modes and $m$ idler modes which should be united in the formalism of the previous sections with the total number of $n=2m$ modes. The fact that in every photon pair the photons are always created in different sets of modes imposes additional constraints on the matrices $\mathbb{H}_0$ and $\mathbb{H}_I$, defining the Gaussian unitary transformation.

To be specific, we again limit our consideration to the case of twin beams discriminated by frequency, with obvious extensions to other cases. Twin beams generation is characterized by phase-matching conditions, providing emission of signal photons in a frequency band $[\omega_s-\Delta\Omega_0,\omega_s+\Delta\Omega_0]$ and idler photons in a frequency band $[\omega_i-\Delta\Omega_0,\omega_i+\Delta\Omega_0]$ such that the frequency gap between the highest idler frequency $\omega_i+\Delta\Omega_0$ and lowest signal frequency $\omega_s-\Delta\Omega_0$ is much higher than the width of the pump spectrum. In this case emission of two photons into one band (signal or idler) is practically impossible, because the sum of their frequencies would be outside of the spectrum of the pump. As consequence, the elements of the matrix $\mathbb{H}_I^{[1]}$, defined by  Eq.~(\ref{HIm}), are almost zero when two indices belong to the same group of modes (signal or idler). With a good degree of approximation we can put them to zero exactly and write the matrix $\mathbf{H}^{[1]}$ defining the generator of the Gaussian transformation in the first-order Magnus expansion as
\begin{equation}\label{HTB}
 \mathbf{H}^{[1]} = \left(\begin{array}{cccc} 0 & 0 & 0 & J \\
                                        0 & 0 & J^T & 0 \\
                                        0 & J^* & 0 & 0 \\
                                        J^\dagger & 0 & 0 & 0
                    \end{array}\right),
\end{equation}
where $J$ is a complex $m\times m$ matrix, known as joint spectral amplitude (JSA) for two photons generated in an elementary nonlinear process. This matrix in general possesses no properties of symmetry.

As indicated above, our consideration here is limited to the first order Magnus approximation and for simplicity we omit the superscripts and the subscripts indicating the first order of Magnus expansion. We adopt the name of ``twin beam Gaussian transformation'' for a unitary transformation $\mathcal{U} = e^{-i\mathcal{H}}$ with a generator $\mathcal{H} = \frac12\mathbf{a}^\dagger \mathbf{H} \mathbf{a}$, where the matrix $\mathbf{H}$ has the structure defined by Eq.~(\ref{HTB}).

The main regularity inherent in twin beams is stated by the following theorem.

\begin{theorem}\label{teo}
All squeezing eigenvalues of a twin beam Gaussian transformation have multiplicities of at least two.
\end{theorem}

\begin{proof}
Applying the SVD to the matrix $-iJ$, we write $-iJ=CR_JD^\dagger$, where $C$ and $D$ are unitary $m\times m$ matrices and $R_J$ is a diagonal $m\times m$ matrix with non-negative elements. Now we can write the $n\times n$ matrix $\mathbb{H}_I$, defined according to Eq.~(\ref{H}), as
\begin{equation}\label{HITB}
 -i\mathbb{H}_I = -i\left(\begin{array}{cc} 0 & J \\ J^T & 0 \end{array}\right) = \mathbb{\tilde V} \mathbb{\tilde R} \mathbb{\tilde V}^T,
\end{equation}
where
\begin{equation}\label{VR}
 \mathbb{\tilde V} = \left(\begin{array}{cc} C & 0 \\ 0 & D^* \end{array}\right),\,\,\,\, \mathbb{\tilde R}=\left(\begin{array}{cc} 0 & R_J \\ R_J & 0 \end{array}\right).
\end{equation}
Therefore the full $2n\times 2n$ generator matrix defined by Eq.~(\ref{HTB}) is
\begin{equation}\label{OBMR1}
 -i\mathbf{H} = \left(\begin{array}{cc} \mathbb{\tilde V} & 0 \\ 0 & \mathbb{\tilde V}^* \end{array}\right) \left(\begin{array}{cc} 0 & \mathbb{\tilde R} \\ -\mathbb{\tilde R} & 0 \end{array}\right) \left(\begin{array}{cc} \mathbb{\tilde V} & 0 \\ 0 & \mathbb{\tilde V}^* \end{array}\right)^\dagger
\end{equation}
and the corresponding symplectic matrix is
\begin{equation}\label{OBMR2}
 \mathbf{S} = e^{-i\mathbf{K}\mathbf{H}} = \left(\begin{array}{cc} \mathbb{\tilde V} & 0 \\ 0 & \mathbb{\tilde V}^* \end{array}\right) \exp\left(\begin{array}{cc} 0 & \mathbb{\tilde R} \\ \mathbb{\tilde R} & 0 \end{array}\right) \left(\begin{array}{cc} \mathbb{\tilde V} & 0 \\ 0 & \mathbb{\tilde V}^* \end{array}\right)^\dagger,
\end{equation}
which corresponds to a decomposition into two-mode squeezers. Note, that this decomposition is not the Bloch-Messiah reduction, since the matrix $\mathbb{\tilde R}$ is not diagonal.

The true Bloch-Messiah reduction can be obtained by an additional rotation. For each pair of conjugated modes the sub-matrix of $\mathbb{\tilde R}$ is proportional to the Pauli matrix $\sigma_x$. The Takagi factorization for this matrix is $\sigma_x = xx^T$, where
\begin{equation}\label{x}
  x = \frac1{\sqrt{2}}\left(\begin{array}{cc} 1 & i \\ 1 & -i \end{array}\right).
\end{equation}
Defining the matrix $\mathbb{X}$ as a direct sum of $m$ matrices $x$ for each pair of conjugated modes, we rewrite Eq.~(\ref{HITB}) as
\begin{equation}\label{HITB2}
 -i\mathbb{H}_I = \left(\begin{array}{cc} C & 0 \\ 0 & D^* \end{array}\right)\mathbb{XRX}^T\left(\begin{array}{cc} C & 0 \\ 0 & D^* \end{array}\right)^T,
\end{equation}
where the diagonal matrix
\begin{equation}\label{RTB}
 \mathbb{R} = \left(\begin{array}{cc} R_J & 0 \\ 0 & R_J \end{array}\right)
\end{equation}
is the matrix of squeezing eigenvalues. Now, defining
\begin{equation}\label{VTB}
 \mathbb{V} = \left(\begin{array}{cc} C & 0 \\ 0 & D^* \end{array}\right)\mathbb{X},
\end{equation}
we interpret Eq.~(\ref{HITB2}) as the Takagi factorization of the squeezing matrix, which is equivalent to the Bloch-Messiah decomposition in the form Eq.~(\ref{BMR1}). It is clear from Eq.~(\ref{RTB}) that each squeezing eigenvalue is found twice in the diagonal of $\mathbb{R}$, and therefore has multiplicity of at least two.
\end{proof}

The columns of matrices $C$ and $D$ are known as modal functions of the signal and idler Schmidt modes of a photon pair generated in the elementary process of photon-pair creation \cite{Law00,Grice01}. It follows from the proof of Theorem~\ref{teo}, that the two squeezing eigenmodes corresponding to the same eigenvalue can be constructed from the Schmidt modes in a straightforward way. Let us denote the $j$th columns of the matrices $C$ and $D$ by $C_j$ and $D_j$ respectively, similar to how it was done in Sec.~\ref{eig} for $V_j$. Both $C_j$ and $D_j$ are $m\times1$ matrices (column vectors) and their corresponding singular value is $r_j$. From Eqs.~(\ref{VTB}), (\ref{x}) we obtain a $m\times2$ matrix containing two squeezing eigenmodes corresponding to the same eigenvalue $r_j$
\begin{equation}\label{VTB2}
 V(r_j) = \left(\begin{array}{cc} C_j & 0 \\ 0 & D_j^* \end{array}\right)x = \frac1{\sqrt{2}}\left(\begin{array}{cc} C_j & iC_j \\ D_j^* & -iD_j^*  \end{array}\right).
\end{equation}
Thus, one squeezing eigenmode can be constructed by concatenating the signal Schmidt mode $C_j$ and the conjugated idler Schmidt mode $D_j$. Another squeezing eigenmode can be constructed in a similar way from $iC_j$ and $iD_j$.

Let us study the uniqueness of the squeezing eigenmodes constructed in this way. For simplicity we suggest that each singular eigenvalue $r_j$ of $-iJ$ has multiplicity 1. Then by the Autonne's uniqueness theorem \cite{HornJohnson} the eigenvectors $C_j$ and $D_j$ are defined up to arbitrary phase $\varphi_j$, i.e. $C_j e^{i\varphi_j}$ and $D_j e^{i\varphi_j}$ are eigenvectors of an alternative SVD of $-iJ$. It means that the general form of Eq.~(\ref{VTB2}) is
\begin{eqnarray}\label{VTB3}
 V(r_j,\varphi_j) &=& \frac1{\sqrt{2}}\left(\begin{array}{cc} C_j e^{i\varphi_j} & iC_j e^{i\varphi_j}\\ D_j^* e^{-i\varphi_j} & -iD_j^*e^{-i\varphi_j}  \end{array}\right) \\\nonumber
 &=& V(r_j,0) \left(\begin{array}{cc} \cos\left(\varphi_j\right) & -\sin\left(\varphi_j\right) \\ \sin\left(\varphi_j\right) & \cos\left(\varphi_j\right)  \end{array}\right).
\end{eqnarray}

It follows from the last expression that the squeezing eigenmodes corresponding to a given eigenvalue with multiplicity 2 are defined up to an orthogonal rotation. This rotation corresponds to shifting the phases of the signal and idler Schmidt modes by the same angle.

The multiplicity of the eigenvalues and the lack of definiteness of the eigenfunctions can be understood from considering just two modes with photon annihilation operators $a$ and $b$. The following identity holds \cite{Milburn84}:
\begin{equation}\label{identity}
 e^{r\left(ab-a^\dagger b^\dagger\right)} =
 e^{\frac{r}2\left(c^2-c^{\dagger2}\right)} e^{\frac{r}2\left(d^2-d^{\dagger2}\right)},
\end{equation}
where the real number $r$ is the squeezing parameter \cite{Kolobov99} and we have introduced two other modes with photon annihilation operators $c = e^{i\pi/4}(a - ib)/\sqrt{2}$ and $d = e^{i\pi/4}(b - ia)/\sqrt{2}$. The left hand side of Eq.~(\ref{identity}) represents the two-mode squeeze operator \cite{Caves85}, while the right hand side of this equation represents a product of two single-mode squeeze operators. Thus, a two-mode squeezed state can be represented by a unitary rotation in the modal space as direct product of two single-mode squeezed states with the same squeezing parameter. An important and less known feature of this representation consists in invariance of the right hand side of Eq.~(\ref{identity}) with respect to orthogonal rotations in the modal space, $c\to c\cos\phi + d\sin\phi$, $d\to -c\sin\phi + d\cos\phi$, where $\phi$ is an arbitrary angle. Such a rotation is similar to that described by Eq.~(\ref{VTB3}).

The practical importance of Theorem~\ref{teo} is connected to its general applicability when the signal and the idler waves are well separated. Being a consequence of a symmetry condition, the multiplicity of the eigenvalues can be used as \emph{a priori} information for fitting and reconstruction of the field state at the measurement stage.

\subsection{Reduction of the Takagi factorization to a spectral decomposition for twin beams}\label{eigenvalue}

We have seen that the Takagi factorization of the squeezing matrix $\Gamma=-i\mathbb{H}_I$ can be obtained by an SVD in a two times smaller space of JSA. An alternative method consists in finding a spectral decomposition for the Hermitian matrix (cf. Theorem 7.3.3. of Ref.~\cite{HornJohnson})
\begin{equation}\label{Gammaa}
 \Gamma_a = \left(\begin{array}{cc} 0 & -iJ \\ iJ^\dagger & 0 \end{array}\right) = \mathbb{U}\Lambda \mathbb{U}^\dagger,
\end{equation}
which we call ``associated squeezing matrix''. It can be shown by direct substitution that the vectors
\begin{equation}\label{pm}
 U_j^{(+)} = \frac1{\sqrt{2}}\left(\begin{array}{c} C_j \\ D_j \end{array}\right),\,\,\,
 U_j^{(-)} = \frac1{\sqrt{2}}\left(\begin{array}{c} C_j \\ -D_j \end{array}\right)
\end{equation}
are eigenvectors of $\Gamma_a$ with the eigenvalues $r_j$ and $-r_j$ respectively. Comparing Eq.~(\ref{pm}) with Eq.~(\ref{VTB2}) we arrive at a conclusion that the modal functions of the squeezing eigenmodes can be obtained by finding the eigenvectors of $\Gamma_a$ and then complex conjugating the idler part and shifting the phase of the vector with the negative eigenvalue by $\pi/2$.

This method is especially fruitful in the case of real matrix $\Gamma=-i\mathbb{H}_I$, where we have a situation already met in Sec.~\ref{eig}. If we consider a  pump pulse which is transform-limited at the center of the crystal, then the squeezing matrix is real (for an appropriately chosen pump phase), as follows from Eq.~(\ref{HIm}) with $z_0=L/2$:
\begin{equation}\label{HIm2}
 \Gamma_{kl}=|\sigma| L\mathcal{E}_{j+l}^0\sinc(\Delta_{kl}L/2).
\end{equation}
The associated squeezing matrix $\Gamma_a$ in this case coincides with $\Gamma$ and is also real symmetric. The eigenvectors $U_j^{(+)}$ and $U_j^{(-)}$ can be chosen real in this case \cite{HornJohnson}. Note that the eigenvectors of a complex spectral decomposition, Eq.~(\ref{Gammaa}) are complex vectors defined up to an arbitrary phase, while the eigenvectors of the real spectral decomposition, Eq.~(\ref{orth}), are real vectors defined up to a sign. It means that by passing from a complex to a real matrix we find the eigenmodes in a more definite and reproducible way.

Indeed, choosing $U_j^{(+)}$ and $U_j^{(-)}$ real, we can omit the complex conjugation of the idler part of these vectors. It follows that $U_j^{(+)}$ and $iU_j^{(-)}$ are the modal functions of the squeezing eigenmodes corresponding to the squeezing eigenvalue $r_j$ with the multiplicity 2. Note that the two modal functions found by this algorithm are almost uniquely defined (up to sign), one of them being purely real, the other one being purely imaginary. Of course, orthogonal rotations in the space of these two functions give us other possible modal functions of the squeezing eigenmodes.

\section{Example of twin beams discriminated by frequency}\label{example}

\subsection{Nondegenerate phase matching: numerical solution}

As an example of the theory developed in the previous sections we consider twin beams generated by PDC in a $2$ mm long sample of BBO (beta barium borate, $\beta$-BaB$_2$O$_4$) crystal. For an angle between the crystal optical axis and the pump wave-vector of $\theta_0=\ang{28.81}$ the type-I phase matching is satisfied for pumping at 397.5 nm, the signal and idler frequency being around 677 nm and 963 nm, respectively. We consider a Gaussian pump pulse, which is transform limited at the center of the crystal with the intensity full width at half maximum $\tau_p=129$ fs, corresponding to a spectral width of $1.8$ nm. The squeezing matrix is given in this case by Eq.~(\ref{HIm2}) and is shown in Fig.~\ref{fig:jsa} as function of detunings $\Omega$ and $\Omega'$. The signal and the idler bands shown in Fig.~\ref{fig:jsa} are well separated, which corresponds to twin beams generation.

\begin{figure}[!ht]
\centering
\includegraphics[width=0.95\linewidth]{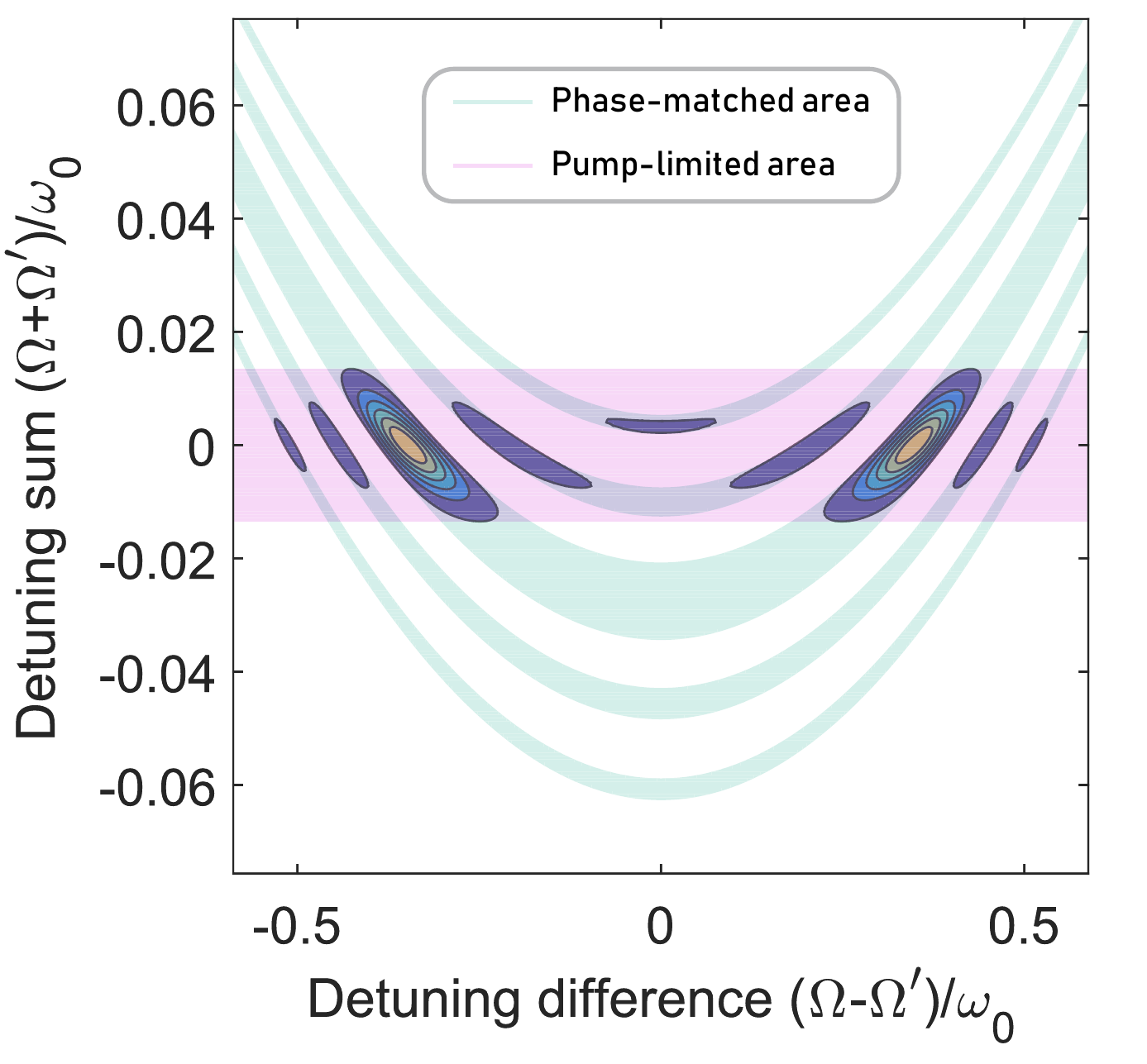}
\caption{The squeezing matrix $\Gamma=-i\mathbb{H}_I$ as a function of detunings $\Omega$ and $\Omega'$ for frequency-nondegenerate PDC. The matrix elements are essentially nonzero at the intersection of the phase-matched area, where the sinc function in Eq.~(\ref{HIm2}) is close to 1, and the pump-limited area, where $|\Omega+\Omega'|$ is within the pump bandwidth. The two regions where the elements of $\Gamma$ are close to maximum correspond to the signal and the idler bands.}
\label{fig:jsa}
\end{figure}

We solve numerically the eigenvalue problem for real symmetric matrix $\Gamma=-i\mathbb{H}_I$, using the Sellmeier equation for the dispersion of the BBO crystal. The squeezing eigenvalues, given by the absolute values of the eigenvalues of $\Gamma$, $r_k=|\lambda_k|$, are shown in Fig.~\ref{fig:values}. We see that each squeezing eigenvalue has multiplicity 2, as stated by Theorem~\ref{teo}, and its logarithm is a linear function of the eigenmode pair number, so that the eigenvalues can be approximated as
\begin{equation}\label{rk}
  r_{1+2l}=r_{2+2l} \approx r_1q_{num}^l
\end{equation}
with $q_{num}=0.8903$ and $l$ running from 0 to $\infty$.
\begin{figure}[!ht]
\centering
\includegraphics[width=0.99\linewidth]{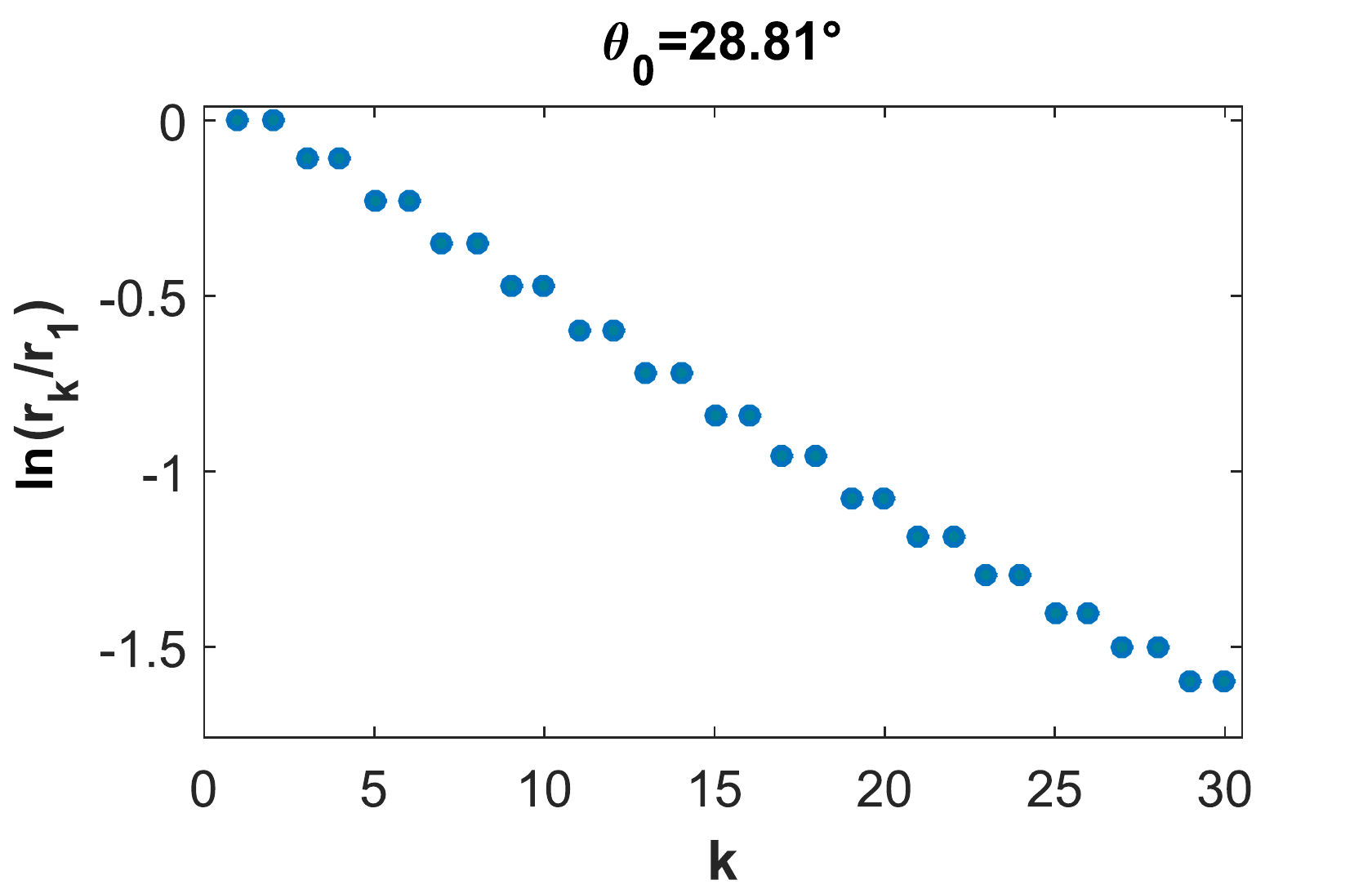}
\caption{Squeezing eigenvalues of twin beams generated in frequency-nondegenerate PDC. For well-separated signal and idler spectral regions all eigenvalues have multiplicity 2 and the logarithm of eigenvalue is a linear function of the mode pair number.}
\label{fig:values}
\end{figure}
\begin{figure}[!ht]
\centering
\includegraphics[width=0.98\linewidth]{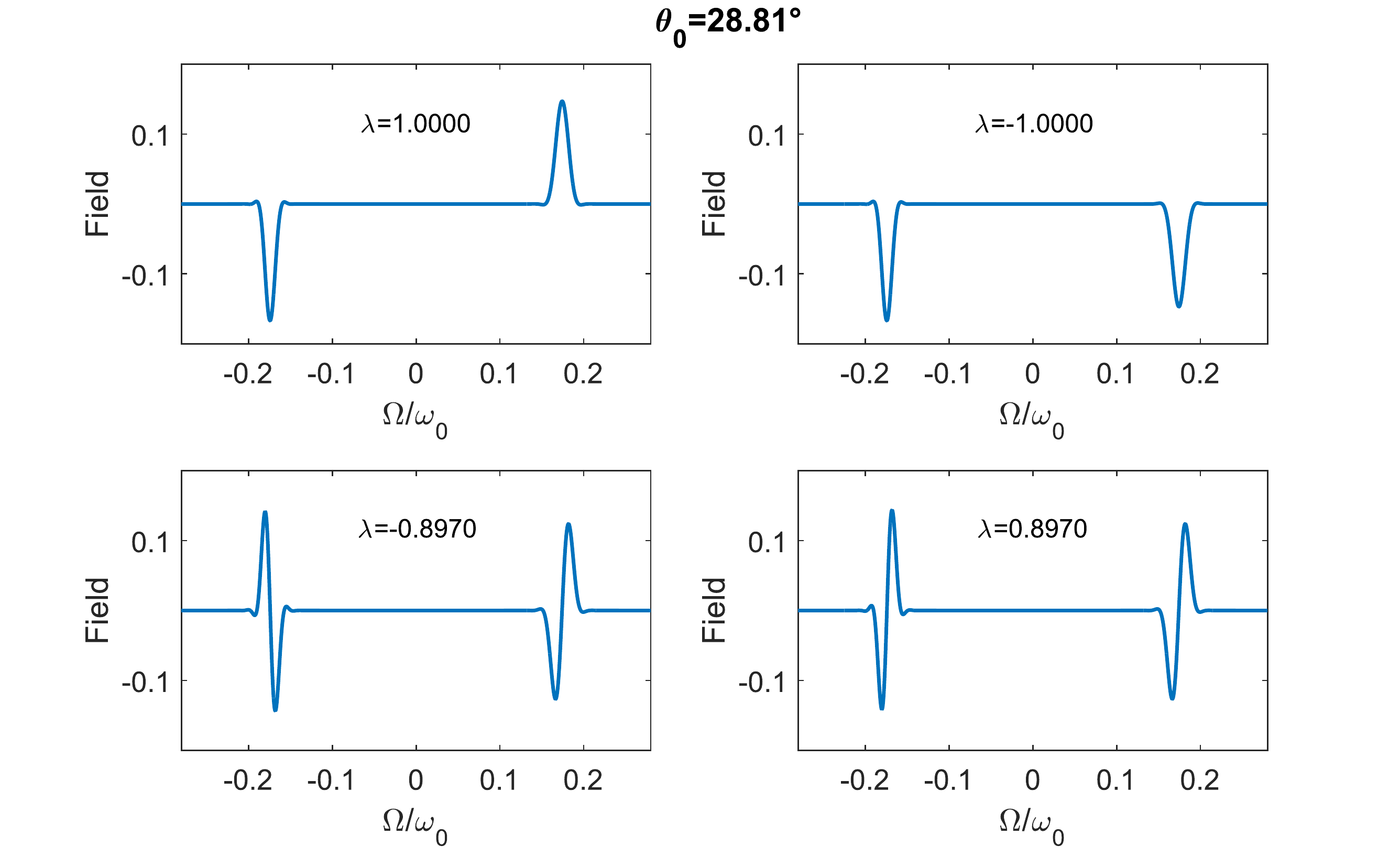}
\includegraphics[width=0.98\linewidth]{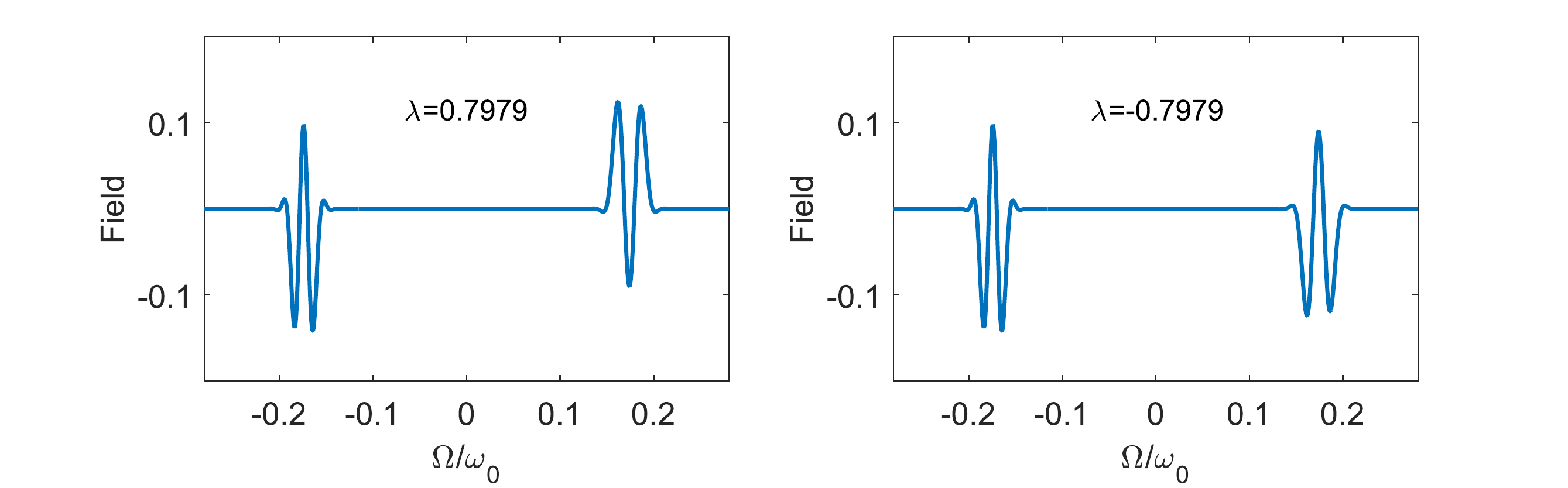}
\caption{Modal functions of the squeezing eigenmodes of twin beams generated in frequency-nondegenerate PDC. The corresponding eigenvalues $\lambda$ of the squeezing matrix $\Gamma$ are shown for each mode. The squeezing eigenvalues are given by $|\lambda|$ and have multiplicity 2, as prescribed by Theorem~\ref{teo}. Each modal function represents a concatenation of two Schmidt modal functions localized in the idler and the signal frequency bands.}
\label{fig:modes}
\end{figure}

The corresponding modal functions of the squeezing eigenmodes are shown in Fig.~\ref{fig:modes}. Modal functions corresponding to the positive eigenvalues of $\Gamma$ are real. Modal functions corresponding to its negative eigenvalues are purely imaginary and their imaginary part is plotted. The structure of the modal functions corresponds to that prescribed by Eq.~(\ref{VTB2}): they are concatenations of Schmidt modal functions for the signal and the idler beams.

The effective number of squeezing eigenmodes can be estimated by a measure similar to the Schmidt number, widely used in the photon-pair generation regime \cite{Law04,Horoshko12,Gatti12}:
\begin{equation}\label{Schmidt}
   K_S = \frac{\left(\sum_{k=1}^{\infty}r_k\right)^2}{\sum_{k=1}^{\infty}r_k^2} = 2\frac{1+q_{num}}{1-q_{num}} \approx 34,
\end{equation}
where we have used the approximate expression for the squeezing eigenvalues, Eq.~(\ref{rk}). We see that as $q_{num}$ approaches unity the effective number of squeezing eigenmodes tends to infinity.

\subsection{Nondegenerate phase matching: analytic Gaussian modeling}

As has been pointed out before, the problem of finding the squeezing eigenvalues and the squeezing eigenmodes can be approached by the SVD of the JSA matrix $J$, rather than by the Takagi factorization of the squeezing matrix $\Gamma$. In this section we follow this approach. The JSA matrix, defined by Eq.~(\ref{HTB}), is given by the top right $m\times m$ block of the matrix $\mathbb{H}_I^{[1]}$, determined by Eq.~(\ref{HIm}). We set $z_0=L/2$, as in the previous section, but here we consider a pump pulse transform-limited at the crystal input face, which is a typical experimental situation, so that at the crystal center it acquires a frequency-dependent phase determined by Eq.~(\ref{E0}).

One advantage of the JSA approach consists in the economy of computational resources: the JSA matrix can be limited only to the spectral regions where the signal and the idler amplitudes are nonnegligible, while the squeezing matrix is mainly filled by zeros in the case of well-separated twin beams. Another advantage is the possibility to find the squeezing eigenvalues and the squeezing eigenmodes approximately by replacing the JSA with a double-Gaussian function and applying the well-known spectral decomposition of a real symmetric double-Gaussian kernel \cite{Grice01,Wasilewski06,Lvovsky07,Patera10,Horoshko12}
\begin{equation}\label{Mehler}
\frac1{\sqrt{\pi}}e^{-\frac{1+q^2}{2(1-q^2)}\left(x^2+y^2\right)+\frac{2q}{1-q^2}xy}
= \sum_{k=0}^\infty pq^k h_k(x) h_k(y),
\end{equation}
where $-1<q<1$, $p=\sqrt{1-q^2}$ and $h_k(x) = \left(2^kk!\sqrt{\pi}\right)^{-\frac12}H_k(x)e^{-x^2/2}$ is the Hermite-Gauss function, $H_k(x)$ being the Hermite polynomial. Equation~(\ref{Mehler}) is obtained by multiplying both sides of the Mehler's formula for Hermite polynomials \cite{Mehler1866} by $e^{-x^2/2-y^2/2}$. Since the Hermite-Gauss functions are orthonormal and complete on the Hilbert space, they represent eigenfunctions of the double-Gaussian kernel, while $pq^k$ are the corresponding eigenvalues.

Though Eq.~(\ref{Mehler}) proved to be highly efficient for finding the eigenmodes and the eigenvalues of a degenerate PDC, in the nondegenerate case one needs an SVD for a complex double-Gaussian kernel, in general not symmetric. Such a decomposition can be obtained in the following form (see proof in Appendix):
\begin{eqnarray}\label{ComplexMehler}
\frac1{\sqrt{\pi}}&&e^{-\frac12\left(\mu x^2+\nu y^2\right)+(\eta+i\xi)xy} = \sqrt{\frac{\tau_1\tau_2}v}\\\nonumber
&&\times \sum_{k=0}^\infty pq^k h_k(\tau_1x) h_k(\tau_2y)e^{i\zeta\tau_1^2x^2 +i\zeta\tau_2^2y^2+i\theta_k},
\end{eqnarray}
where $\mu,\nu,\eta$ and $\xi$ are real numbers, satisfying the relations $\mu,\nu>0$, $\sqrt{\mu\nu}>|\eta|$, necessary and sufficient for this kernel to be square-integrable. The parameters in the right-hand side of Eq.~(\ref{ComplexMehler}) are $\tau_1=\sqrt{uv/\nu}$, $\tau_2 = \sqrt{uv/\mu}$, $\zeta=\eta\xi/(2uv)$, $p=\sqrt{1-q^2}$, and
\begin{eqnarray}\label{parametersM}
q &=& \sqrt{\frac{u-v}{u+v}},
\end{eqnarray}
where $u=\sqrt{\mu\nu+\xi^2}$ and $v=\sqrt{\mu\nu-\eta^2}$. The phase $\theta_k$ is determined in Appendix. Note that the factor $v^{-\frac12}$ in the right-hand side of Eq.~(\ref{ComplexMehler}) is the Hilbert-Schmidt norm of the kernel in its left-hand side. In Eq.~(\ref{Mehler}) this norm is equal to 1. Another difference in the structure of two decompositions is the sign of the parameter $q$. In Eq.~(\ref{Mehler}) this parameter can be negative, while in Eq.~(\ref{ComplexMehler}) it is always positive and is complemented by the phase $\theta_k$. In the limiting case $\xi=0$, $\mu=\nu=\sqrt{1+\eta^2}$, $\theta_k=\pi k (1-\eta/|\eta|)/2$, Eq.~(\ref{ComplexMehler}) reproduces Eq.~(\ref{Mehler}).

Equation~(\ref{ComplexMehler}) shows that the left-singular functions of a complex double-Gaussian kernel are chirped Hermite-Gauss functions $\psi_k(x)=\sqrt{\tau_1}h_k(\tau_1x)e^{i\zeta\tau_1^2x^2+i\theta_k/2}$. It is straightforward to verify that these functions create a complete orthonormal set on the Hilbert space. The right-singular functions of this kernel are obtained from the left-singular ones by repacing $\tau_1$ by $\tau_2$.

In this section for the sake of analytic treatment we write the JSA function as a kernel depending on two continuous frequencies $\Omega_1$ and $\Omega_2$, but similar approach applies to the matrices obtained by passing to discrete frequency modes. To write the JSA in a double-Gaussian form, we make four following approximations.

The first approximation is the quadratic approximation for the dispersion law, which is a good approximation for a not too broadband PDC \cite{Caspani10,Horoshko12,Gatti12}. It has been applied to the pump in Eq.~(\ref{kp}). Applying it to the downconverted light results in limiting the Taylor series of the phase mismatch function $\Delta(\Omega_1,\Omega_2) =k_p(\Omega_1+\Omega_2)-k(\Omega_1)-k(\Omega_2)$ to the terms up to the second order in both frequencies:
\begin{eqnarray}\label{Delta}
\Delta(\Omega_1,\Omega_2) &\approx& \Delta_0 + (k_{p0}'-k_0')\Omega_+ \\\nonumber
&+& \frac14 (2k_{p0}''-k_0'')\Omega_+^2 - \frac14 k_0'' \Omega_-^2,
\end{eqnarray}
where $\Delta_0=k_{p0}-2k_0$ and we have introduced new variables $\Omega_\pm=\Omega_1\pm\Omega_2$. Here we denote by $k_{0}$, $k_{0}'$ and $k_{0}''$ the value and the two derivatives of $k(\Omega)$ at $\Omega=0$, as we did for the pump in Sec.~\ref{eig}.

The second approximation consists in disregarding the third term in the right hand side of Eq.~(\ref{Delta}) with respect to the second one, which can be done if $|\Omega_+|\ll4|k_{p0}'-k_0'|/|2k_{p0}''-k_0''|$. For the considered example this means $|\Omega_+|\ll 1.04\omega_0$, while $|\Omega_+|$ is limited by the pump bandwidth $\Omega_p=2\sqrt{\log2}/\tau_p\approx0.0091\omega_0$, which means that this approximation is well justified. In this approximation the line in the frequency space corresponding to the perfect phase matching is represented by a parabola
\begin{eqnarray}\label{parabola}
\Delta_0 + (k_{p0}'-k_0')\Omega_+ - \frac14 k_0'' \Omega_-^2 = 0.
\end{eqnarray}
This parabola, corresponding to the middle of the phase-matched area in Fig.~\ref{fig:jsa}, intersects the $\Omega_+=0$ line at two points $\Omega_-=\pm2\Omega_s$, where $\Omega_s=\sqrt{\Delta_0/k_0''}$ is the central detuning of the signal beam, $-\Omega_s$ being the central detuning of the idler one. In the vicinity of the point $\Omega_+=0$, $\Omega_-=2\Omega_s$ we introduce new variables $\delta\Omega_+=\Omega_+$, $\delta\Omega_-=\Omega_--2\Omega_s$, corresponding to small deflections from the frequencies where the signal-idler coupling is maximal. Further, we disregard $\delta\Omega_-^2$ compared to $\Omega_s\delta\Omega_-$, which is the third approximation we make. This approximation is always valid for twin beams well-separated in frequency, when the bandwidth of one beam is much less than the frequency distance between the central frequencies of the beams.

In these three approximations the JSA function is obtained from Eqs.~(\ref{HIm}), (\ref{E0}) and (\ref{kp}) as
\begin{eqnarray}\nonumber
 J(\Omega_1,\Omega_2)=&g&\exp\left(-\frac{\delta\Omega_+^2}{2\Omega_p^2}\right) \sinc\left(\frac{\tau_d\delta\Omega_+-\tau_s\delta\Omega_-}2\right)\\\label{J}
 &\times& e^{i\tau_{pd}\delta\Omega_+/2 + i\tau_{ps}^2\delta\Omega_+^2/4},
\end{eqnarray}
where $g=\sigma LE_0e^{ik_{p0}L/2}$ is the coupling constant, $E_0$ being the peak pump amplitude, and we have introduced  four characteristic times: the absolute group delay time of the pump $\tau_{pd}=k_{p0}'L$, the pump spread time $\tau_{ps}=\sqrt{k_{p0}''L}$, the relative pump-signal group delay time $\tau_d=(k_{p0}'-k_0')L$, and $\tau_s=\sqrt{k_0''\Delta_0}L$. The latter time is proportional to the spread time of the signal pulse during its passage through the crystal $\sqrt{k_{0}''L}$ \cite{Horoshko12,Gatti12}, but also depends on the phase-mismatch at degeneracy $\Delta_0$. It is straightforward to find that in the vicinity of the point $\Omega_+=0$, $\Omega_-=-2\Omega_s$ the signal-idler coupling is determined by the function $J(\Omega_2,\Omega_1)$, which is the transposition of the JSA defined by Eq.~(\ref{J}).

The fourth and the final approximation consists in making a replacement $\sinc(x)\approx e^{-x^2/(2\sigma_0^2)}$, where $\sigma_0=1.61$ is chosen so that these two functions have the same width at half-maximum \cite{Grice01,Wasilewski06,Lvovsky07,Patera10,Horoshko12}. As a result, JSA in the variables $\delta\Omega_1=\Omega_1-\Omega_s$, $\delta\Omega_2=\Omega_2+\Omega_s$ takes the form $J(\Omega_1,\Omega_2)=g\tilde J(\Omega_1,\Omega_2)e^{i\tau_{pd}\left(\delta\Omega_1+\delta\Omega_2\right)/2 + i\tau_{ps}^2\left(\delta\Omega_1^2+ \delta\Omega_2^2\right)/4}$, where
\begin{eqnarray}\label{J2}
 \tilde J(\Omega_1,\Omega_2) &=&\exp\left(-\frac{(\delta\Omega_1+\delta\Omega_2)^2}{2\Omega_p^2}+\frac{i}2\tau_{ps}^2\delta\Omega_1\delta\Omega_2\right)\\\nonumber &\times&\exp\left(-\frac{[(\tau_d-\tau_s)\delta\Omega_1+(\tau_d+\tau_s)\delta\Omega_2]^2} {8\sigma_0^2}\right)
\end{eqnarray}
is the part of the kernel with a non-separable phase. This function is shown (in its absolute value) in Fig.~\ref{fig:jsaGauss}, where it is compared to the exact one.
\begin{figure}[!ht]
\centering
\includegraphics[width=0.95\linewidth]{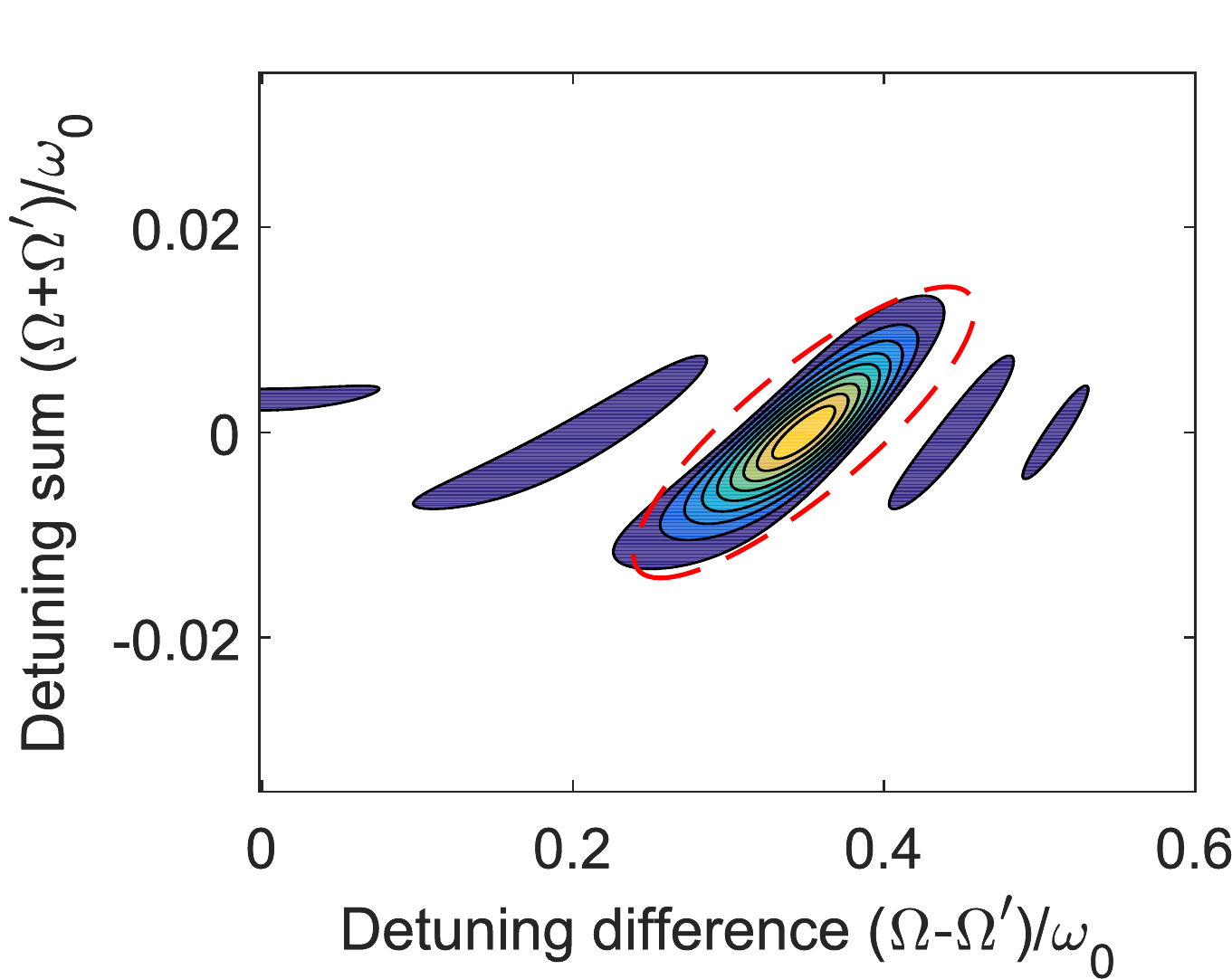}
\caption{JSA as a function of the detunings $\Omega$ and $\Omega'$ (contour map) and its double-Gaussian approximation (red dashed line).}
\label{fig:jsaGauss}
\end{figure}

Comparing Eq.~(\ref{J2}) to Eq.~(\ref{ComplexMehler}) and accepting for simplicity that $\arg(-ig)=0$, we find the modal functions of the signal and the idler Schmidt modes:
\begin{eqnarray}\label{C}
 C_k(\Omega_1) &=& \sqrt{\tau_1}h_k\left(\tau_1\delta\Omega_1\right) e^{i\tau_{pd}\delta\Omega_1+i\zeta_1\tau_1^2\delta\Omega_1^2},\\\nonumber D_k(\Omega_2) &=& \sqrt{\tau_2}h_k\left(\tau_2\delta\Omega_2\right) e^{-i\tau_{pd}\delta\Omega_2-i\zeta_2\tau_2^2\delta\Omega_2^2},
\end{eqnarray}
where the parameters $\tau_1$ and $\tau_2$ are defined as in Eq.~(\ref{ComplexMehler}) with $\xi=\tau_{ps}^2/2$ and
\begin{eqnarray}\nonumber
\mu&=&\frac1{\Omega_p^2}+\frac{(\tau_d-\tau_s)^2}{4\sigma_0^2},
\\\label{physparam}
\nu&=&\frac1{\Omega_p^2}+\frac{(\tau_d+\tau_s)^2}{4\sigma_0^2},
\\\nonumber
\eta&=&-\frac1{\Omega_p^2}-\frac{\tau_d^2-\tau_s^2}{4\sigma_0^2},
\end{eqnarray}
while the signal and the idler chirp rates are $\zeta_1=\xi(\nu+\eta)/(2uv)$ and $\zeta_2=\xi(\mu+\eta)/(2uv)$.

For the considered example $\tau_1=48$ fs, $\tau_2=60$ fs, $\zeta_1=0.0086$, $\zeta_2=-0.0063$, and $q=0.8681$. The modal function of the signal Schmidt mode $C_1(\Omega_1)$ is shown in Fig.~\ref{fig:Schmidt} together with the modal function obtained by a numerical SVD of JSA for two different crystal lengths. We see that the analytic solution follows closely the numerical one for a shorter crystal, with slight differences for a longer one.

\begin{figure}[!ht]
\centering
\includegraphics[width=0.95\linewidth]{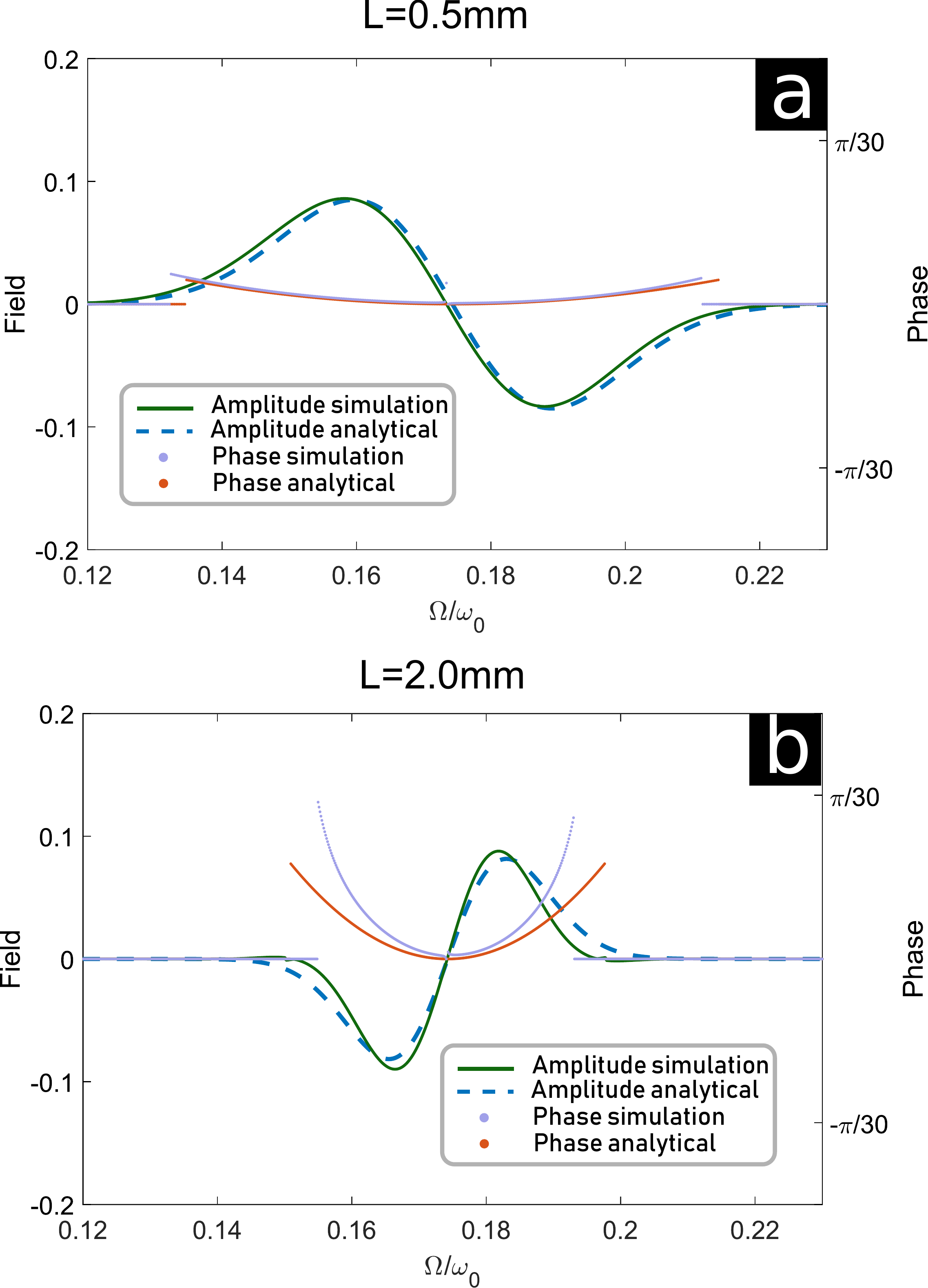}
\caption{Amplitude and phase of the modal function of the (properly delayed) signal Schmidt mode $C_1(\Omega_1)=h_1\left(\tau_1\delta\Omega_1\right) e^{i\zeta_1\tau_1^2\delta\Omega_1^2}$ obtained by Gaussian modeling and by numerical SVD of JSA for two BBO crystal lengths: $L=0.5$ mm in (a) and $L=2.0$ mm in (b). A longer crystal corresponds to a narrower bandwidth. The chirp rate increases with length but for the considered example (BBO pumped at 397.5 nm) is negligible at any crystal length.}
\label{fig:Schmidt}
\end{figure}

We note also, that the variation of the phase of the modal function due to its chirp within the signal frequency band is much less than $\pi$ and therefore negligible for the considered example. This fact justifies the approach of the previous section, where the pump chirp was disregarded from the beginning.

It is interesting to analyze the influence of the crystal length $L$ on the shape of the Schmidt modes. As typical for PDC, growing crystal length results in a more restrictive phase matching and to narrower signal and idler bandwidths. This is clearly seen in Fig.~\ref{fig:Schmidt}. The absolute value of the chirp increases for a longer crystal, however, the maximal variation of the phase remains small because of decreased bandwidth. The dimensionless chirp parameters $\zeta_1$ and $\zeta_2$ allow us to estimate the importance of considering the complex squeezing matrix instead of a real one. In the considered example they are much less that 1. Asymptotically with growing $L$ they tend to length-independent constants, determined by the dispersive properties of the crystal material only, as can be seen from  their definitions and Eq.~(\ref{physparam}).

We observe that the idler characteristic time $\tau_2$ is 25\% higher than the signal characteristic time $\tau_1$. This explains the asymmetry of the signal and idler Schmidt modes in Fig.~\ref{fig:modes}. The eigenvalues shown in Fig.~\ref{fig:values} correspond to the geometric progression $q_{num}^n$, with $q_{num}=0.8903$, which is slightly higher than the value of $q$ given above by the Gaussian modeling. This and other slight differences between the numerical and the analytical solutions are caused by some of the four approximations we made.

The treatment of this section demonstrates the effectiveness of Gaussian modeling for the Schmidt modes, based on the complex Mehler's formula.

\subsection{Twin beams close to degeneracy}
To illustrate the disappearance of the features peculiar to twin beams, now we consider a different angle between the crystal optical axis and the pump wave-vector of $\theta_0=\ang{29.18}$, for which the PDC is close to be frequency degenerate. The squeezing matrix is shown in Fig.~\ref{fig:jsa2}. In this case the signal and the idler areas are not well separated and we expect deflections from the theory developed above.
\begin{figure}[!ht]
\centering
\includegraphics[width=0.99\linewidth]{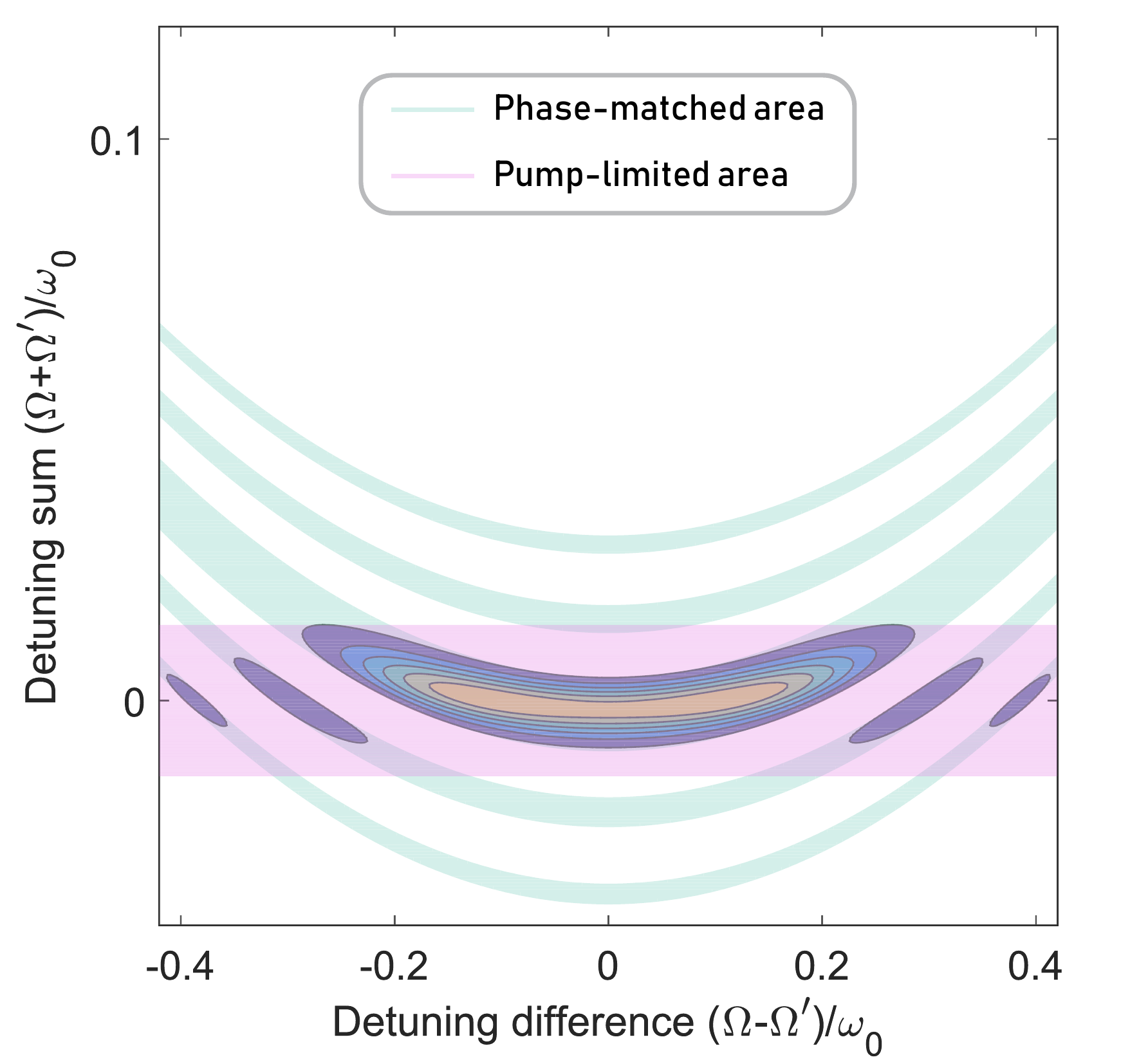}
\caption{The squeezing matrix $\Gamma$ as function of the signal and idler detunings $\Omega$ and $\Omega'$ respectively for type-I PDC close to degeneracy. The signal and the idler bands are not separated in this case.}
\label{fig:jsa2}
\end{figure}

The squeezing eigenvalues are shown in Fig.~\ref{fig:values2}. Only the four first eigenvalues show multiplicity 2, predicted by Theorem~\ref{teo}.
\begin{figure}[!ht]
\centering
\includegraphics[width=0.99\linewidth]{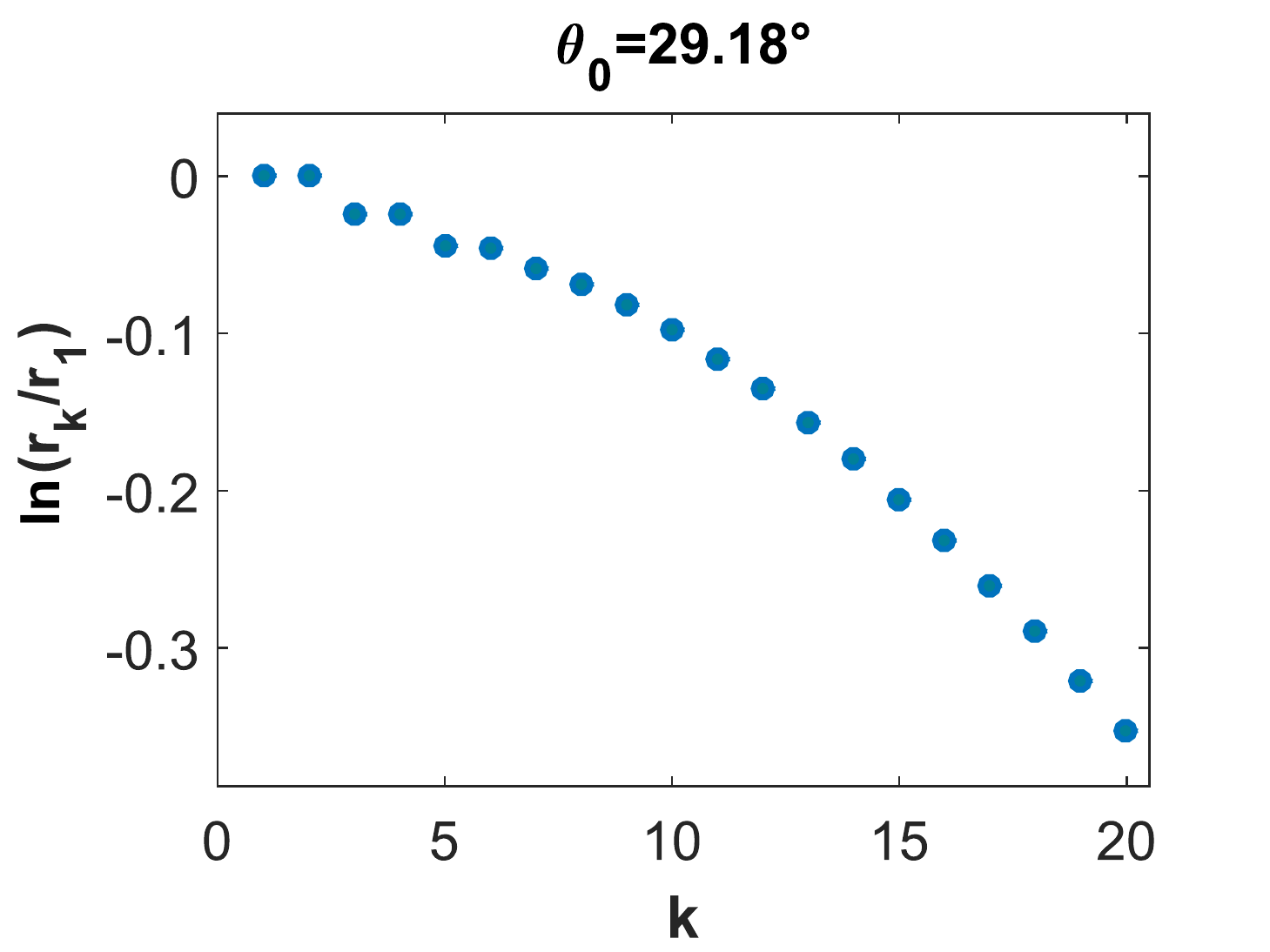}
\caption{Squeezing eigenvalues of twin beams generated in PDC close to degeneracy in frequency. Only four first eigenvalues create pairs, as required by Theorem~\ref{teo}. At higher mode number the emission of two signal or two idler photons in one elementary act becomes significant and the Theorem~\ref{teo} is not valid any more. }
\label{fig:values2}
\end{figure}

The modal functions of the corresponding squeezing eigenmodes are shown in Fig.~\ref{fig:modes2}. Only the four first modal functions resemble concatenations of local Hermite-Gauss modes. At higher mode numbers the deflection of the transformation generator matrix  from the form of Eq.(\ref{HTB}) becomes significant and Theorem~\ref{teo} loses its validity. Physically it means that the emission of two signal or two idler photons in one elementary act of photon-pair creation becomes significant as we approach the degenerate regime by varying the phase-matching conditions. We conjecture that the eigenvalues cease to be multiple at the mode number $k$ for which the Schmidt modal functions calculated separately for the signal and the idler beams start to overlap in the area around the origin. For the considered case it is $k=5$ and $k=6$, see the two bottom plots in Fig.~\ref{fig:modes2}, where the signal and the idler parts of the modal function start to deflect from the shape of the Hermite-Gauss function of the corresponding (second) order.
\begin{figure}[!ht]
\centering
\includegraphics[width=\linewidth]{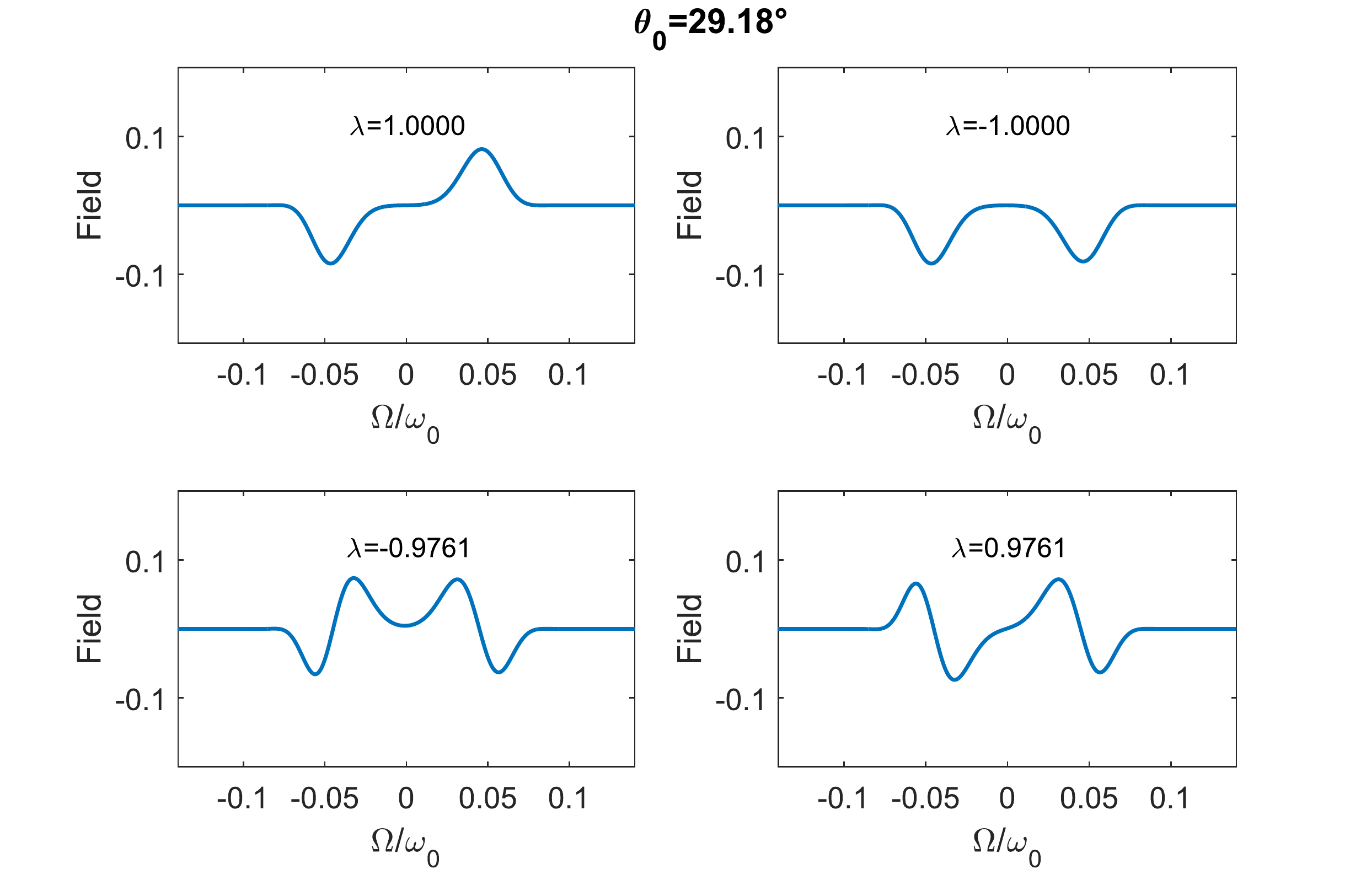}
\includegraphics[width=\linewidth]{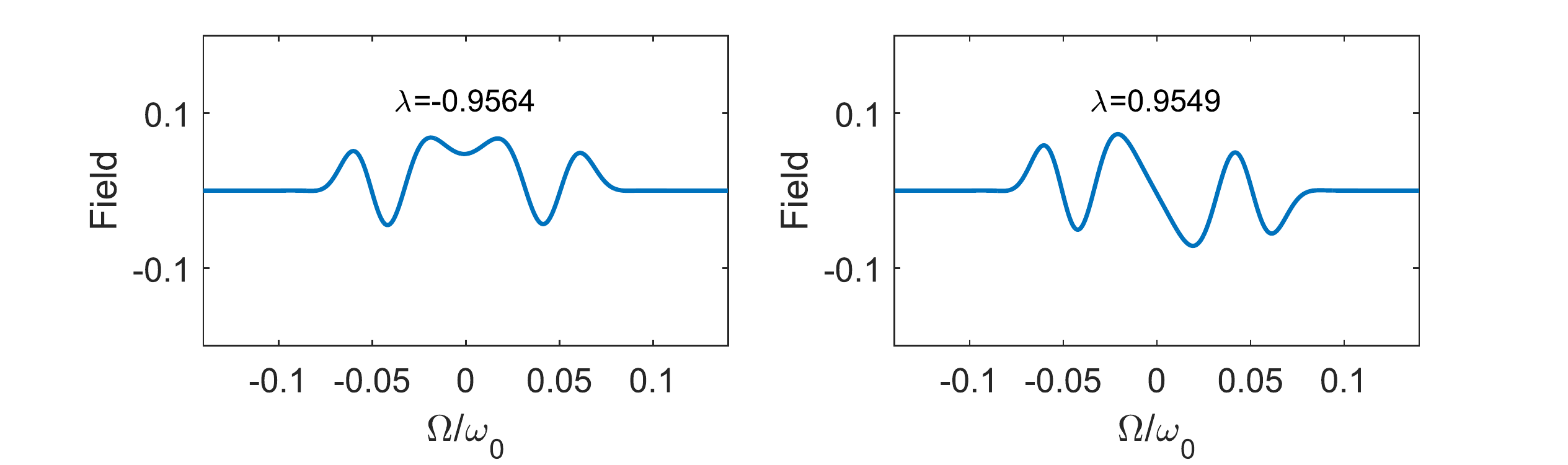}
\caption{Modal functions of the squeezing eigenmodes of twin beams generated in PDC close to degeneracy in frequency. The corresponding eigenvalues of the squeezing matrix $\Gamma$ are shown for each mode. The four first modal functions can be considered as concatenations of the Hermite-Gauss functions in the signal and the idler bands. However, the two last modal functions deflect from this rule around the origin. For these and higher modes the twin beam analysis is not valid any more.}
\label{fig:modes2}
\end{figure}

\section{Conclusions}\label{conclusions}
We have applied the formalism of Bloch-Messiah reduction to the parametric downconversion of light in the case where pulsed twin beams are generated. We have shown how the squeezing eigenvalues and the squeezing eigenmodes can be obtained in the case of moderate squeezing, where the solution of the wave equation is obtained in the first order of Magnus expansion. For this case we have proven a fundamental result: all the squeezing eigenvalues of twin beams have multiplicity at least two. As consequence, the modal functions of the squeezing eigenmodes are not unique and defined up to an orthogonal rotation in the space of two eigenmodes related to the same eigenvalue. We established two methods for avoiding the ambiguity in the definition of the squeezing eigenmodes: (i) reducing the Takagi factorization of the squeezing matrix to the spectral decomposition of an associated Hermitian matrix and (ii) tailoring the squeezing eigenmodes from the Schmidt modes of the signal and the idler beams obtained by SVD of the JSA matrix.

These general results have been illustrated by an example of twin beams discriminated by frequency. For the case of good separation of the signal and the idler spectra we have found the multiplicity two of all eigenvalues, predicted by the general theory. We have also shown that the Schmidt modal functions of the signal and the idler beams can be modeled with very good precision by replacing the JSA by a complex double-Gaussian function. The modal functions of the squeezing eigenmodes in this case are two chirped Hermite-Gauss functions in the signal and the idler frequency bands.

There are several important extensions of the present work. The considered example corresponds to rather narrowband twin beams discriminated by frequency. It would be interesting to apply the developed formalism to ultrabroadband fields with almost constant spectrum, for example, those generated in aperiodically poled crystals \cite{Horoshko13,Horoshko17,Perina18,Chekhova18}. In the narrowband regime a full treatment of 3D spatio-temporal modes is possible in the framework of complex analytic modeling developed here. It should reveal correlations between the spatial and temporal degrees of freedom of entangled twin beams. Another important direction of future research is related to twin beams discriminated by polarization or wave-vector direction. In the latter case two beams can be frequency-degenerate and their correlations can be converted to single-mode squeezing by means of interference on a beam splitter. This approach has high potential for producing highly multimode pulsed cluster states in a form of entangled frequency combs, similar to continuous-wave cluster states successfully produced in the last years \cite{Yokoyama13}.

\section*{Acknowledgments}

D.B.H and M.I.K. thank G. Patera and T. Lipfert for helpful comments. F.A. is grateful to F. Corsi for useful discussions. This work was supported by the European Union's Horizon 2020 research and innovation programme under grant agreement No 665148 (QCUMbER). F.A. acknowledges financial support from the French National Research Agency Project No ANR-17-CE24-0035 VanQuTe.

\appendix*
\section{Proof of the complex Mehler's formula}

Here we prove a relation, which is equivalent to Eq.~(\ref{ComplexMehler}) and is obtained from it by letting $x'=\tau_1x$, $y'=\tau_2y$, $\eta'=\eta/\sqrt{\mu\nu}$, $\xi'=\xi/\sqrt{\mu\nu}$. Omitting the primes for simplicity we obtain the following identity
\begin{eqnarray}\label{ComplexMehler2}
&&\frac1{\sqrt{\pi}}e^{-\frac1{2w}\left(x^2+y^2\right)+\frac1w(\eta+i\xi)xy} \\\nonumber
&&= \sqrt[4]{(1+\xi^2)} pe^{i\theta_0}\sum_{k=0}^\infty \left(q e^{i\theta}\right)^k h_k(x) h_k(y)e^{i\zeta\left(x^2+y^2\right)},
\end{eqnarray}
which can be viewed as the Takagi factorization of a complex \emph{symmetric} double-Gaussian kernel. Here $w=\sqrt{(1+\xi^2)(1-\eta^2)}$, $\zeta=\eta\xi/(2w)$, $p=\sqrt{1-q^2}$, $q=\sqrt{(1-q_0)/(1+q_0)}$, where $q_0=\sqrt{(1-\eta^2)/(1+\xi^2)}$. The angles $\theta$ and $\theta_0$ are found below. Note that the factor $\sqrt[4]{(1+\xi^2)}$ in the right-hand side of Eq.~(\ref{ComplexMehler2}) is the Hilbert-Schmidt norm of the kernel in its left-hand side.

Let us denote the kernel in the left-hand side of Eq.~(\ref{ComplexMehler2}) by $K(x,y)$. This kernel corresponds to some integral operator $\mathcal{K}$. The standard method of finding the singular functions of this operator consists in solving the eigenvalue problems for the Hermitian operators $\mathcal{K}\mathcal{K}^\dagger$ and $\mathcal{K}^\dagger\mathcal{K}$. However, the phases of the singular functions cannot be determined in this way, since the eigenfunction is defined up to a unitary rotation (phase in one dimension), while the singular function in a Takagi factorization is defined up to an orthogonal rotation (sign in one dimension). For a full Takagi factorization we apply a more powerful method of generating function.

The generating function for the Hermite polynomials reads
\begin{eqnarray}\label{Generating}
&&\Phi_0(x,t)=e^{2xt-t^2} = \sum_{k=0}^\infty H_k(x) \frac{t^k}{k!}.
\end{eqnarray}

Now consider the following function
\begin{eqnarray}\label{Generating2}
&&\Phi(x,t)=\int_{-\infty}^{\infty}K(x,y)\Phi_0(y,t)e^{-y^2/2-i\zeta y^2}dy,
\end{eqnarray}
where the last factor under the integral is the complement of the Hermite polynomial to the singular function in Eq.~(\ref{ComplexMehler2}) (up to normalization). Differentiating the integrand $k$ times we obtain
\begin{eqnarray}\label{GeneratingDiff}
\left.\frac{\partial^k\Phi(x,t)}{\partial t^k}\right|_{t=0} = N_k\int_{-\infty}^{\infty}K(x,y)h_k(y)e^{-i\zeta y^2}dy,\\\nonumber
\end{eqnarray}
where $N_k=\sqrt{2^kk!\sqrt{\pi}}$ is the normalization constant of the Hermite-Gauss function.

On the other hand, Eq.~(\ref{Generating2}) represents an integral of a complex Gaussian function of $y$, which can be taken analytically:
\begin{eqnarray}\label{Generating3}
&&\Phi(x,t)=p_c
e^{-x^2/2+i\zeta x^2} e^{2q_cxt-q_c^2t^2},
\end{eqnarray}
where the complex numbers $q_c$ and $p_c$ are defined as
\begin{eqnarray}\label{qc}
q_c = \frac{\eta+i\xi}{1+w+i\eta\xi},\\\label{pc}
p_c = \sqrt{\frac{2w}{1+w+i\eta\xi}}.
\end{eqnarray}
Note that $p_c^2+q_c^2=1$. The last factor in the right-hand side of Eq.~(\ref{Generating3}) is exactly $\Phi_0(x,q_ct)$. Differentiating it $k$ times we obtain
\begin{eqnarray}\label{GeneratingDiff2}
&&\left.\frac{\partial^k\Phi(x,t)}{\partial t^k}\right|_{t=0} =N_kp_cq_c^kh_k(x)e^{i\zeta x^2}.
\end{eqnarray}

Comparing Eq.~(\ref{GeneratingDiff}) with Eq.~(\ref{GeneratingDiff2}) we conclude that the function $\Psi_k(y)=h_k(y)e^{i\zeta y^2+i(\theta_0+k\theta)/2}$, where $\theta_0=\arg(p_c)$ and $\theta=\arg(q_c)$, is the right-singular function of the kernel $K(x,y)$ with the singular value $|p_cq_c^k|$. It is not hard to find that $|q_c|=q$ and $|p_c|=p\sqrt[4]{(1+\xi^2)}$. Since the considered kernel is symmetric, its left-singular function is $\Psi_k(x)$, which concludes the proof.

\bibliography{Twin-beams-BM2} 

\end{document}